\documentclass[12pt]{article}
\usepackage{hyperref}
\usepackage{enumerate}
\usepackage{amsthm}
\usepackage{amsmath}
\usepackage{amssymb}
\usepackage{mathrsfs}
\usepackage{mathtools}
\usepackage{verbatim}
\usepackage[capitalise]{cleveref}
\usepackage{relsize}
\usepackage{xspace}
\usepackage{nicematrix}
\usepackage{multicol}
\usepackage{caption}
\usepackage{subcaption}
\usepackage{tabularx}

\usepackage{fullpage}

\allowdisplaybreaks

\theoremstyle{definition}
\newtheorem{definition}{Definition}[section]
\theoremstyle{plain}
\newtheorem{lemma}[definition]{Lemma}
\newtheorem{theorem}[definition]{Theorem}
\newtheorem{corollary}[definition]{Corollary}
\newtheorem{conjecture}{Conjecture}
\theoremstyle{remark}

\newtheorem{observation}[definition]{Observation}

\newcommand{\Oh}{\mathcal{O}}
\newcommand{\ostar}{\Oh^*}
\DeclareMathOperator*{\poly}{Poly}

\usepackage{macros}
\title{Tight Algorithm for Connected Odd Cycle Transversal Parameterized by Clique-width}
\author{Narek Bojikian\hspace{2cm}Stefan Kratsch}

\begin{document}
\maketitle
\thispagestyle{empty}
\begin{center}
    \textbf{Abstract}
\end{center}
\begin{abstract}
 \small
 \noindent Recently, Bojikian and Kratsch [2023] have presented a novel approach to tackle connectivity problems parameterized by clique-width ($\cw$), based on counting small representations of partial solutions (modulo two). Using this technique, they were able to get a tight bound for the \Stp problem, answering an open question posed by Hegerfeld and Kratsch [ESA, 2023]. We use the same technique to solve the \Coctp problem in time $\ostar(12^{\cw})$. We define a new representation of partial solutions by separating the connectivity requirement from the 2-colorability requirement of this problem. Moreover, we prove that our result is tight by providing SETH-based lower bound excluding algorithms with running time $\ostar((12-\epsilon)^{\lcw})$ even when parameterized by linear clique-width. This answers the second question posed by Hegerfeld and Kratsch in the same paper.
\end{abstract}

\section{Introduction}

Parameterized complexity is a branch of complexity theory that studies the fine-grained complexity of a problem, by analyzing the dependence of the running time on different measures of input in addition to its size. Specifically, for NP-hard problems, one seeks algorithms with polynomial dependence on the input size. While not likely achievable in general, this becomes more realistic, when one allows exponential dependence on the introduced measurement $t$, hence yielding polynomial time algorithms when restricting the input to instances with constant value of $t$. In general, one seeks an algorithm with running time $f(t)n^c$ for a constant value of $c$ and some computable function $f$. Problems admitting an algorithm with such running time define the complexity class FPT, that builds the core of the study of parameterized complexity.

Some of the most prominent measures in this aspect, aside from the size of the output, have been structural parameters of graphs, namely parameters that reflect how well-structured a graph is. These measures, mostly bounded in sparse graphs or in graphs that admit a lot of symmetries, are usually accompanied by a corresponding structural decomposition of the graph, that describes the structure of the graph in a recursive manner. Such decompositions usually allow to build the graph using a handful of structure-preserving operations. Therefore, one can usually describe dynamic programming algorithms over these decompositions to solve different problems in FPT running time.

The most notable parameter in this area has probably been treewidth, a measure that sparks interest in different areas of theoretical computer science, not only parameterized complexity. Vaguely defined by the smallest value $k$, such that a given graph can be decomposed into single vertices by recursively removing separators of size at most $k$. FPT algorithms with single exponential dependence on the treewidth of the graph have long been known for different problems. However, one would still try to improve the exponential dependence on the parameter, by improving the base of the exponent in the running time. In 2010, Lokshtanov et al.\ \cite{DBLP:journals/corr/abs-1007-5450, DBLP:journals/talg/LokshtanovMS18} provided the first tight lower bounds for problems parameterized by treewidth, proving that known bases for some benchmark problems cannot be improved assuming the strong exponential time hypothesis (SETH), a widely used hypothesis for lower bounds in theoretical computer science, that, informally speaking, conjectures that the naive solution of the SAT problem, by testing all possible assignments, is essentially optimal. Lokshtanov et al.\ \cite{DBLP:journals/talg/LokshtanovMS18} provide reductions from the $d$-SAT problem to the given problem, resulting in a graph with bounded treewidth, which would imply the lower bound. This has inspired a series of SETH based lower bounds for structural parameters \cite{DBLP:conf/soda/CurticapeanM16, DBLP:journals/corr/abs-2210-10677/EsmerFMR22, DBLP:journals/jgaa/GeffenJKM20, DBLP:journals/siamdm/Lampis20, DBLP:conf/iwpec/HegerfeldK22, DBLP:conf/esa/HegerfeldK23}.

However, a special class of problems, called connectivity problems, vaguely defined by imposing connectivity constraints on the solution (e.g.\ \Cvcp), or by constraining the connectivity of the rest of the graph (e.g.\ \Fvsp) resisted single exponential running time algorithms for a long time. A major obstacle had been keeping track of different connectivity patterns in the boundary of the graph, representing different partial solutions in this graph. That seemed to necessitate a factor of $k^{\Oh(k)}$ in the running time. However, in 2011, Cygan et al.\ \cite{DBLP:journals/corr/abs-1103-0534/CyganNPPRW11, DBLP:journals/talg/CyganNPPRW22} came around this problem by introducing the \cnc technique. They used it, among others, to count the number of connected solutions of a problem modulo two in single exponential time. By using the isolation lemma \cite{DBLP:journals/combinatorica/MulmuleyVV87}, they showed that one can solve the underlying decision problem with high probability. Moreover, they proved that the resulting base of the exponent in the running time is tight assuming SETH.

Relying on the \cnc technique, tight bounds for connectivity problems were also derived for different parameters such as pathwidth \cite{DBLP:journals/jacm/CyganKN18}, cutwidth \cite{DBLP:conf/stacs/BojikianCHK23}, modular treewidth \cite{DBLP:conf/wg/HegerfeldK23} and clique-width \cite{DBLP:conf/esa/HegerfeldK23,DBLP:journals/corr/abs-2307-14264/BojikianK23}. In this paper, we focus on Clique-width, a graph parameter defined by the smallest value $k$, such that a given graph can be built recursively from the disjoint union of $k$-labeled graphs, by allowing bicliques between vertices of different labels, and relabeling vertices. Clique-width is a more general parameter than treewidth, since cliques have clique-width at most $2$, and unbounded treewidth, while graphs of treewidth $k$ have bounded clique-width (at most $2^{\Oh(k)}$) \cite{CourcelleO00, DBLP:journals/siamcomp/CorneilR05}.

The first single exponential upper bounds for connectivity problems parameterized by clique-width, to the best of our knowledge, were given by Bergougnoux and Kant\'e \cite{DBLP:journals/tcs/BergougnouxK19}. They provide algorithms that, as mentioned in the paper, are optimal under the ETH Hypothesis, which only excludes algorithms with running time $2^{o(k)}\poly(n)$, but does not show the optimal base of the exponent. The first tight results for connectivity problems were given by Hegerfeld and Kratsch \cite{DBLP:conf/esa/HegerfeldK23}, where they provided algorithms, based on the \cnc technique, and matching SETH-based lower bounds for the \Cvcp problem and the \Cdsp problem parameterized by clique-width. They left open however, the tight complexity of other connectivity problems such as \Stp and the \Coctp problem. A major obstacle towards finding the the optimal base of the exponent, as they mention, is a gap between the rank of a matrix defined by the interactions of partial solutions of the underlying problem (called \emph{the consistency matrix}), and its largest triangular submatrix. Except for a handful examples (\textsc{Chromatic Number}[$\ctw$] \cite{DBLP:journals/tcs/JansenN19}, and \Fvsp[$\cw$] \cite{DBLP:conf/stacs/BojikianCHK23}), the rank of such matrices has usually been used directly to derive the optimal running time, while current lower bound techniques usually match the size of a largest triangular submatrix.

Very recently, Bojikian and Kratsch \cite{DBLP:journals/corr/abs-2307-14264/BojikianK23} closed this gap for the first of these problems, by providing a faster algorithm for the \Stp problem, based on small representation of connectivity patterns. They showed that one can count the number of these representations modulo two, in time proportional to the size of a largest triangular submatrix of the corresponding matrix. Using the isolation lemma, they solved the underlying decision problem in the same running time, matching known lower bounds for this problem. In order to do so, they introduced the notion of action-sequences to distinguish different representations of the same solution.

\subparagraph*{Our results.}

We follow the approach of Bojikian and Kratsch to provide an optimal (modulo SETH) algorithm for the \Coctp problem parameterized by clique-width. We show that this problem can be solved in time $\ostar(12^k)$, for $k$ the clique-width of the underlying graph. Namely, we prove the following theorem:

\begin{theorem}\label{theo:upper-bound}
    There exists a one-sided error Monte-Carlo algorithm (only false negatives), that given a graph $G$ together with a $k$-expression of $G$, and a positive integer $\target$, solves the \Coctp\ problem in time $\ostar(12^k)$, and outputs the correct answer with probability at least one half.
\end{theorem}

Our algorithm counts the number of representations of connected solutions modulo 2. We use the isolation lemma \cite{DBLP:journals/combinatorica/MulmuleyVV87}, to isolate a single representation of a single solution with high probability. We make use of  fast convolution methods over power lattices \cite{DBLP:journals/talg/BjorklundHKKNP16,DBLP:conf/esa/HegerfeldK23,DBLP:conf/stacs/BojikianCHK23} and other algebraic techniques.

Moreover, we prove that our running time is tight under SETH. In order to achieve this we follow a long line of SETH based lower bounds for structural parameters \cite{DBLP:journals/talg/LokshtanovMS18,DBLP:journals/jacm/CyganKN18,DBLP:conf/stacs/GroenlandMNS22,DBLP:journals/talg/CyganNPPRW22,DBLP:conf/stacs/BojikianCHK23,DBLP:journals/siamdm/Lampis20}, by providing a reduction from the $d$-SAT problem preserving the value of the parameter.

Our lower bound actually excludes faster algorithms for problems parameterized by \emph{linear clique-width} \cite{DBLP:journals/tcs/AdlerK15,DBLP:journals/dam/HeggernesMP12,DBLP:journals/tcs/HeggernesMP11a,DBLP:journals/tcs/GurskiW05}, a parameter lower bounded by clique-width. This shows that the claimed base is optimal for both clique-width and linear clique-width assuming SETH.
A similar approach of more restrictive parameterization for lower bounds (pathwidth) compared to the upper bound parameter (treewidth) can bee seen in \cite{DBLP:journals/talg/CyganNPPRW22}.
Using this reduction, we prove the following theorem:

\begin{theorem}\label{theo:lower-bound}
    Given a graph $G$ of linear clique-width at most $k$, assuming SETH, there exists no algorithm that solves the \Coctp problem in time $\ostar\big((12-\varepsilon)^k\big)$ even when a linear $k$-expression of $G$ is provided with $G$.
\end{theorem}

\subparagraph*{Further related work.}

Clique-width was first studied by Courcelle and Olariu \cite{CourcelleO00} showing that problems expressible in $\msoOne$ can be solved in FPT time when parameterized by clique-width. However, while the algorithm has linear dependence on the input size, it could have a non-elementary dependence on the parameter value $k$, and hence, is probably not tight for most relevant problems. Espelage et al.\ \cite{DBLP:conf/wg/EspelageGW01} have provided XP algorithms for some problems not expressible in $\msoOne$ logic. Nonetheless, some problems have been proven to be para-NP-hard when parameterized by clique-width, and hence probably do not admit a polynomial time algorithm on graphs of bounded clique-width, such as the disjoint paths problem \cite{DBLP:journals/tcs/GurskiW06}.

Tight bounds for problems parameterized by treewidth were first found by Lokshtanov et al.\ \cite{DBLP:journals/talg/LokshtanovMS18}, and were followed by series of tight bounds for structural parameters. Iwata and Yoshida \cite{DBLP:conf/esa/IwataY15} have proven that vertex cover has the same base parameterized by treewidth and clique-width, proving one of the earliest tight bounds for problems parameterized by clique-width. Further tight bounds were obtained for counting perfect matchings, graph coloring and other problems \cite{DBLP:conf/soda/CurticapeanM16,DBLP:journals/dam/KatsikarelisLP19,DBLP:journals/siamdm/Lampis20,DBLP:conf/icalp/GanianHKOS22}.

As mentioned earlier, the \cnc technique, introduced by Cygan et al.\ \cite{DBLP:journals/talg/CyganNPPRW22}, mainly counts the number of partial solutions (or other predefined objects) modulo two. Hence, resulting algorithms for the corresponding decision problems are inherently randomized. Bodlaender et al.\ \cite{DBLP:journals/iandc/BodlaenderCKN15} provided a rank-based approach to solve connectivity problems parameterized by treewidth in deterministic single exponential running time using Gauss elimination. However, this derandomization came at the cost of a factor $\omega$ (the matrix multiplication factor) in the exponent of the running time, resulting in non-optimal running times. Both techniques have proven useful to derive faster algorithms for connectivity problems over different parameters. While the latter approach was used to provide single exponential algorithms for connectivity problems parameterized by clique-width \cite{DBLP:journals/tcs/BergougnouxK19}, the former was used to prove tight bounds for connectivity problems parameterized by different parameters such as cutwidth \cite{DBLP:conf/stacs/BojikianCHK23} and modular treewidth \cite{DBLP:conf/wg/HegerfeldK23}.

We use fast convolution methods in order to get a tight algorithm. Such methods are usually used to combine different partial solutions in different parts of the graph along a boundary set efficiently. Fast convolution over the subset lattice was introduced by Björklund et al.\ \cite{DBLP:conf/stoc/BjorklundHKK07} giving rise to the fast subset convolution technique. Using this technique, the authors showed how to solve the \Stp problem parameterized by the number of terminals more efficiently. Fast subset convolution, and variations thereof, have since turned a standard technique to design dynamic programming algorithms for different problems over a tree decomposition of a graph \cite{DBLP:conf/esa/RooijBR09,DBLP:journals/talg/CyganNPPRW22} (mainly to handle join nodes). A more general form of convolution (over join-lattices) was used by van Rooij \cite{DBLP:conf/csr/Rooij21} to provide faster algorithms for the class of \textsc{$[\sigma,\rho]$-Dominating Set} problems parameterized by treewidth.

Björklund et al.\ \cite{DBLP:journals/talg/BjorklundHKKNP16} showed that one can compute the so-called join-product, a standard convolution over lattices, more efficiently, if the underlying lattice admits a small number of so-called irreducible elements. Hegerfeld and Kratsch \cite{DBLP:conf/esa/HegerfeldK23}, and Bojikian and Kratsch \cite{DBLP:journals/corr/abs-2307-14264/BojikianK23} used this result to design optimal algorithms (modulo SETH) for connectivity problems parameterized by clique-width.

\subparagraph*{Structure of this work.}

In \cref{sec:preliminaries} we provide preliminaries and introduce some notation. In \cref{sec:pats} we define the notion of patterns, representation and pattern operations. In \cref{sec:algo} we present the main result of this paper. We describe the dynamic programming algorithm, and prove its running time. In \cref{sec:acseq} we define action-sequences, and use them to prove the correctness of the algorithm. Finally, in \cref{sec:lb} we present our SETH based lower bound, proving the tightness of our algorithm. We conclude in \cref{sec:conclusion} with open problems and final thoughts.

\section{Preliminaries}\label{sec:preliminaries}

\subparagraph{Graphs, clique expressions and odd cycle transversals.}

% - clique-width
For a natural number $k$, we denote by $[k] = \{1,2,\dots k\}$ the set of natural numbers smaller than or equal to $k$, and $[k]_0 = [k]\cup\{0\}$.
In this work we deal with undirected graphs only. A \emph{labeled graph} is a graph $G=(V, E)$ together with a \emph{labeling function} $\lab\colon V\rightarrow \mathbb{N}$. We usually omit the function $\lab$ and assume that it is given implicitly when defining a labeled graph $G$. We define a \emph{clique-expression} $\mu$ as a well-formed expression that consists only of the following operations on labeled graphs:
\begin{itemize}
    \item \emph{Introduce} vertex $i(v)$ for $i\in\mathbb{N}$. This operation constructs a graph containing a single vertex and assigns label $i$ to this vertex.
    \item The \emph{union} operation $G_1 \union G_2$. The resulting graph consists of the disjoint union of the labeled graphs $G_1$ and $G_2$.
    \item The \emph{relabel} operation $\relabel{i}{j}(G)$ for $i,j \in \mathbb{N}$. This operations changes the labels of all vertices labeled $i$ in $G$ to the label $j$.
    \item The \emph{join} operation $\add i j (G)$ for $i,j\in\mathbb{N}, i\neq j$. The constructed graph results from $G$ by adding all edges between the vertices labeled $i$ and the vertices labeled $j$, i.e.
    \[\add{i}{j}(G) = \Big(V, E \cup \big\{\{u, v\}\colon\lab(u)=i \land \lab(v)=j\big\}\Big).\]
\end{itemize}
We denote the graph resulting from a clique-expression $\mu$ by $G_{\mu}$, and the constructed labeling function by $\lab_{\mu}$. We associate with a clique-expression $\mu$ a syntax tree $\syntaxtree_{\mu}$ in the natural way, and associate with each node $x\in V(\syntaxtree)$ the corresponding operation. Let $\nodes_{\mu} = V(\syntaxtree_{\mu})$. We omit $\mu$ when clear from the context.

Let $x\in \nodes$. The subtree rooted at $x$ induces a subexpression $\mu_x$. We define $G_x = G_{\mu_x}$, $V_x = V(G_x)$, $E_x=E(G_x)$, $\lab_x = \lab_{\mu_x}$, $\syntaxtree_x = \syntaxtree_{\mu_x}$, and $\nodes_x = \nodes_{\mu_x}$.
 We also define the set $\cnodes_x$ as the set of all introduce nodes in $\syntaxtree_x$, and $\wnodes_x$ as the set of all introduce and join nodes in $\syntaxtree_x$.
Given a set $S\subseteq V$, we define $S_x = S\cap V_x$. Also for a mapping $g:S\rightarrow U$ to some set $U$, we denote by $g_x$ the mapping $g_{|S_x}$, the restriction of $g$ to vertices in $S\cap V_x$.

We call a clique-expression $\mu$ \emph{linear}, if for each union node $x$ with children $x_1$ and $x_2$, i.e.\ $\mu_x = \mu_{x_1}\union \mu_{x_2}$, it holds that $\mu_{x_2}$ is an introduce operation. We say that $G$ is $k$-labeled, if it holds that $\lab(v)\leq k$ for all $v\in V$. We say that a clique-expression $\mu$ is a \emph{$k$-expression} if $G_x$ is a $k$-labeled graph for all $x\in \nodes$. We define the (linear) clique-width of a graph $G$, denoted by $\cw(G)$ ($\lcw(G)$ for linear clique-width) as the smallest value $k$ such that there exists a (linear) $k$-expression $\mu$, with $G_{\mu}$ isomorphic to $G$.
We can assume without loss of generality, that any given $k$-expression defining a graph $G=(V,E)$ uses at most $O(|V|)$ union operations, and at most $O(|V|k^2)$ unary operations~\cite{DBLP:journals/tcs/BergougnouxK19, CourcelleO00}.

% COCT
A mapping $g:V\rightarrow [q]$, for $q\in\mathbb{N}$, is a \emph{proper $q$-coloring} of $G$, if and only if for all $\{u,v\}\in E$ it holds that $g(u)\neq g(v)$. A graph $G$ is \emph{$q$-colorable}, if it admits a proper $q$-coloring. A set $S\subseteq V$ is an \emph{odd cycle transversal} of $G$, if $G\setminus S$ is $2$-colorable. In this case, we call a proper $2$-coloring $g$ of $G\setminus S$ a \emph{witness}.
We define the \Coctp problem as:

\begin{center}
\fbox{
    \parbox{.75\textwidth}{
        \Coctp
        \begin{itemize}
            \item \textbf{Input:} A graph $G$ and a positive integer $\target \in \mathbb{N}$.
            \item \textbf{Question:} Does $G$ admit an odd cycle transversal $S$ of size at most $\target$, such that $S$ induces a connected subgraph of $G$?
        \end{itemize}
}}\end{center}

\subparagraph{Algebraic notation, sets, functions and lattices.}

We denote by $\Delta$ the symmetric difference of two sets, i.e.\
\[
  S_1\Delta S_2 = (S_1\setminus S_2)\cup (S_2\setminus S_1).
\]
Clearly, $\Delta$ is an associative operation. Hence, for a family of sets $\mathcal{F}=\{S_1, \dots, S_n\}$, the following notion is well-defined
\[
    \bigdelta\limits_{S\in\mathcal{F}}S = S_1\Delta S_2 \Delta \dots \Delta S_n.
\]
% Mappings
A \emph{weight function} $\omega$ is a mapping from some ground set $U$ to some field (usually $\mathbb{Z}$). Given a set $S\subseteq U$, we define
\[\omega(S) = \sum\limits_{u\in S}\omega(u).\]
Let $f:U^k\rightarrow U$ be a $k$-ary mapping, for some integer $k$. If $f$ is not explicitly defined as a weight function, we define $f(S_1,\dots S_k)$ for some sets $S_1,
\dots S_k \subseteq U$ as follows
\[
f(S_1,\dots S_k) = \{f(u_1,\dots, u_k)\colon u_1\in S_1, \dots, u_k\in S_k\}.
\]
Similarly, for $S\subseteq U$, we define
\[
f^{-1}(S) = \{(u_1,\dots, u_k)\colon f(u_1,\dots,u_k)\in S\}.
\]
Finally, we define the operation $\exop{f}\colon(\pow{U})^k\rightarrow \pow{U}$, where for $S_1,\dots,S_k \subseteq U$ we have
\[
    \exop{f}(S_1,\dots,S_k) = \bigdelta\limits_{(u_1,\dots,u_k)\in S_1\times\dots\times S_k} \big\{f(u_1,\dots,u_k)\big\},
\]
and call it the exclusive version of $f$.
Let $f:U_1\rightarrow U$ and $g:U_2\rightarrow U$, be two mappings such that $U_1\cap U_2=\emptyset$. We define the disjoint union of $f$ and $g$ as $f\dot\cup g\colon U_1\cup U_2\rightarrow U$ where 
\[
  f\dot\cup g  (u) = 
  \begin{cases}
  f(u) &\colon u \in U_1,\\
  g(u) &\colon u \in U_2. 
  \end{cases}
\]

%lattices
A \emph{lattice} $(\lat, \leq)$ is a partial ordering over a set $U$, such that for all $u,v\in \lat$ there exist
\begin{itemize}
	\item an element $w\in \lat$ called the \emph{join element} ($w=u\lor v$), where $u\leq w$, $v\leq w$ and for all $w'$ with $u\leq w'$ and $v\leq w'$ it holds that $w\leq w'$,
	\item and an element $z \in \lat$ called the \emph{meet element} ($z = u\land v$), where $z\leq u$, $z\leq v$, and for all $z'$ such that $z'\leq u$ and $z'\leq v$ it holds that $z'\leq z$.
\end{itemize}
We call two lattices isomorphic, if the underlying orderings are isomorphic to each other. We slightly abuse the notation, and refer to a lattice $(\lat, \leq)$ either by $\lat$ or $\leq$ only, if the other part is clear from the context.
For a lattice $(\lat, \leq)$, we define the set of \emph{join-irreducible} elements $\lat_{\lor} \subseteq \lat$ as the set of elements $u \in \lat$ where for all $v,w \in \lat$ with $u=v\lor w$ it holds that $u=v$ or $u=w$.
We say that a lattice $(\lat, \leq)$ is given in the \emph{join representation} if we are given the elements of $\lat$ as $\Oh(\log |\lat| )$-bit strings, the elements of $\lat_{\lor}$, and an algorithm $\mathcal{A}_{\lat}$ that computes the join $u\lor v$ for $u\in \lat$ and $v\in\lat_{\lor}$.

Let $A,B\in \mathbb{F}^{\lat}$ be two vectors over some field $\mathbb{F}$. We define the \emph{$\lor$-product} $A\otimes B\in \mathbb{F}^\lat$ of $A$ and $B$ where for $z\in \lat$ it holds that
\[
	A\otimes B[z] = \sumstack{x,y\in\lat\\ x\lor y = z} A[x]\cdot B[y].
\]

\begin{lemma}[Folklore]\label{lem:v-prod-iso}
        Let $\lat$ and $\lat'$ be two isomorphic lattices, and let $\tau$ be a lattice isomorphism between $\lat$ and $\lat'$, where both $\tau$ and its inverse can be computed in time $g(n)$. If $\lor$-product over $\lat'$ can be computed in time $f(n)$, then $\lor$-product over $\lat$ can be computed in time $\Oh(f(n)+ |\lat|\cdot g(n))$.
\end{lemma}

\subparagraph*{Conventions.}
Along this work, let $G = (V, E)$ be the input graph and let $\target \in \mathbb{N}$ be the target value in the \Coctp instance. Let $n = |V|$ and $m = |E|$. Let $\mu$ a $k$-expression of $G$ for some value $k \in \mathbb{N}$, and $\syntaxtree$ be the corresponding syntax tree. Let $r$ be the root of $\syntaxtree$. A \emph{partial solution} at $x\in\nodes$ is a set of vertices $S\subseteq V_x$.

For a function $f$, we denote by $\ostar(f(k))$ the running time $\Oh(f(k)\cdot \poly(n)),$ where $\poly(n)$ denotes some arbitrary polynomial in $n$. When we say that some function is bounded polynomially, we always mean polynomially in $n$ (and hence, in the size of the input).

We define the ground set $U = [k]$, and call its elements \emph{labels}. Let $U_0 = [k]_0$. 
Let $\W \in \mathbb{N}$ be some fixed value that we choose later, bounded polynomially by $|V(\syntaxtree)|$ (and hence, by $n$). Let $\weightf:\nodes\times [4]\rightarrow [\W]$ be a weight function chosen independently and uniformly at random from $[\W]^{\nodes\times [4]}$. From now on, we use the indices $\budget$ and $\weight$ to iterate over the ranges $[\target]_0$ and $[|\nodes| \cdot \W]_0$. We skip defining these notations repeatedly to avoid redundancy.

\section{Connectivity patterns}\label{sec:pats}

In this section we define patterns, as structures that represent connectivity in labeled graphs. We use the definition of patterns and some of their properties provided in \cite{DBLP:journals/corr/abs-2307-14264/BojikianK23}.

\begin{definition}
    A \emph{pattern} $p$ is a subset of the power-set of $U_0$ with exactly one set $Z_p$ containing the element $0$, called the \emph{zero-set} of $p$. We denote by $\Pat$ the family of all patterns. We denote by $\lbs(p) \subseteq U$ the set of all labels that appear in at least one set of $p$, and by $\sing(p) \subseteq U$ the set of all labels that appear as singletons in $p$ (we exclude the element $0$ from both sets.) We define $\inc(p)=\lbs(p)\setminus\sing(p)$, and call it the set of \emph{incomplete labels} of $p$.
\end{definition}

We sometimes write 
\[
    [i^1_1\dots i^1_{k_1}\boldsymbol{,} \cdots \boldsymbol{,}  i^{\ell}_1\dots i^{\ell}_{k_\ell}]    
\]
to denote the pattern
\[
    \big\{\{i^1_1, \dots, i^1_{k_1}\}, \dots, \{i^{\ell}_1,\dots, i^{\ell}_{k_{\ell}}\}\big\},    
\]
 where we use square brackets to enclose the pattern, we do not use any separator between the elements of the same set, and separate different sets with comas. We only use this concise way of writing when we can represent each label by a single symbol, and hence, the corresponding pattern is well-defined. We sometimes ignore the element $0$ if it appears as a singleton. For example, for two variables $i$ and $j$, both $[ij]$ and $[ij, 0]$ denote the pattern $\{\{0\}, \{i,j\}\}$, while $[10,12]$ denotes the pattern $\{\{0,1\}, \{1, 2\}\}$. The pattern $\{\{0,10\}\}$ does not admit such a short writing, since it contains an element consisting of more than one symbol.

Now we define pattern operations. These operations will allow us to build connectivity patterns corresponding to different partial solutions recursively over the nodes of $\syntaxtree$ (see \cite{DBLP:journals/corr/abs-2307-14264/BojikianK23} for details).

\begin{definition}\label{def:patops}
    Let $p,q \in \Pat$, and let $r \in \Pat$ be the pattern resulting from each of the following operations. We define
\begin{center}
    \begin{tabular}{ l p{.82\linewidth} }
        Join:&$r = p\join q$. Let $\sim_I \subseteq p\times q$ be the relation defined over $p\cup q$ where for $S \in p$ it holds that $S\sim_I S'$ if $S' \in q$ and $S\cap S' \neq \emptyset$. We define the equivalence relation $\reachable[p,q]$ (omitting $p$ and $q$ when clear from the context) as the reflexive, transitive and symmetric closure of $\intersects$. Let $\mathcal{R}$ be the set of equivalence classes of $\reachable$, then we define the pattern $r = \{\bigcup_{S \in P} S\colon P\in \mathcal{R}\}$, as the unions of the sets in each equivalence class of $\reachable$.\\
        Relabel:&$r = p_{i\rightarrow j}$, for $i,j\in U$, where $r$ results from $p$ by replacing $i$ with $j$ in each set of $p$ that contains $i$.\\
        Union:&$r=p\punion q$, where $r = (p\setminus \{Z_p\}) \cup (q \setminus \{Z_q\}) \cup \{Z_p\cup Z_q\}$.\\
    \end{tabular}
\end{center} 
\end{definition}

\begin{definition}\label{def:patadd}
    For $i,j\in [k]$ with $i\neq j$, we define the operation $\patadd_{i,j}$ over $\Pat$ as
    \[
        \patadd_{i,j}p = 
        \begin{cases}
            p\join [ij] &\colon \{i,j\}\subseteq \lbs(p),\\
            p           &\colon\text{otherwise}.
        \end{cases}    
    \]
\end{definition}

Now we define the consistency relation over patterns. This relation indicates whether two labeled subgraphs join into a connected graph by connecting all vertices of the same label between them.

\begin{definition}\label{def:consistency}
    Let $p,q\in \Pat$. We say that $p$ and $q$ are consistent (denoted by $p\sim q$), if for $r = p\join q$ it holds that $r = \{Z_r\}$ contains the zero-set only, i.e.\ the join operation merges all sets of both patterns into a single set.
\end{definition}

\begin{definition}    
We identify two families of special patterns, namely the family of \emph{complete patterns} $\Cp$ and the family of \emph{$CS$-patterns} $\CSP$. We say that a pattern $p$ is \emph{complete} ($p\in \Cp$) if it holds that $\sing(p) = \lbs(p)$, i.e.\ each label that appears in this pattern, appears as a singleton as well. We call a complete pattern a \emph{$CS$-pattern} ($p\in\CSP$), if it consists only of a zero-set and singletons.
\end{definition}

\begin{lemma}\label{lema:zero-only-complete-consistent}
    The pattern $[0]$ is the only complete pattern consistent with the pattern $[0]$.
\end{lemma}
\begin{proof}
    This follows from the fact that $p\join[0] = p$ for all patterns $p$, and that any other complete pattern contains at least one singleton different from $\{0\}$.
\end{proof}

\begin{definition}\label{def:forget-fix}
    Let $p\in \Pat$. For $i\in U$, we define the operations $\fix(p,i)$ and $\forget(p,i)$ as follows: If $i \in \lbs(p)\setminus \sing(p)$, we set
    \begin{alignat*}{2}
        &\fix(p,i) &&=p\cup \big\{\{i\}\big\},\\
        &\forget(p,i) &&=\big\{S\setminus\{i\}\colon S\in p\big\}.
    \end{alignat*}
    Otherwise, we set 
    \[
        \fix(p,i) = \forget(p, i) = p.
    \]
    We say that $p' = \fix(p, i)$ results from $p$ by \emph{fixing} the label $i$, and $p'' = \forget(p,i)$ results from $p$ by \emph{forgetting} the label $i$.
\end{definition}

\begin{definition}\label{def:actions}
    We define \emph{actions} $\action\colon \Pat \times [4] \rightarrow (\Cp\cup\uparrow)$, as pattern operations, where we denote by $\uparrow$ the value "undefined". We define $\action$ as follows: Given a pattern $p \in \Pat$, if $p\in \Cp$, we define $\action(p, \ell)= p$, for all $\ell\in [4]$.
    For $\inc(p)=\{i\}$, we define 
    \begin{itemize}
        \item $\action(p, 1) = \fix(p, i)$,
        \item $\action(p, 2) = \forget(p, i)$.
    \end{itemize}
    We set $\action(p, 3) = \action(p, 4) = \uparrow$ in this case. Finally, if $\inc(p)=\{i, j\}$ for $i < j$, we define 
    \begin{itemize}
        \item $\action(p, 1) = \fix(\fix(p, i), j)$,
        \item $\action(p, 2) = \forget(\fix(p, i), j)$,
        \item $\action(p, 3) = \fix(\forget(p, i), j)$,
        \item $\action(p, 4) = \forget(\forget(p, i), j)$.
    \end{itemize}
    In all other cases we set $\action(p, \ell) = \uparrow$ for all $\ell\in [4]$.
\end{definition}

\begin{observation} \label{lem:pats-closed-ops}
    Both families $\Cp$ and $\CSP$ are closed under the operations $\punion$, and $_{i\rightarrow j}$ for all values of $i, j\in [k]$. On the other hand, given $p\in \Cp$ and $i,j\in[k]$, it holds that either $\patadd_{i,j}p = p$ (if $\{i,j\}\not\subseteq \lbs(p)$), or $\inc(\patadd_{i,j}p) = \{i,j\}$ otherwise.
\end{observation}

Finally, we define for a complete pattern $p$ the family $\parrep(p)\subseteq \CSP$ that, as we shall see in \cref{lem:parrep-parity-rep}, correctly determines whether $p$ is consistent with another complete pattern $p'$.

\begin{definition}\label{def:parrep}
    Let $p\in\Cp$. We define $\parrep(p) \subseteq \CSP$ as follows: 
    Let $P_0 = \{p\}$, and let $S_1,\dots S_r$ be all non-singleton sets of $p$ different from $Z_p$.
    We define $P_i$ for $i\in[r]$ recursively as follows:
    \[
    P_i = \bigdelta\limits_{q\in P_{i-1}} \Big\{ \big(q\setminus \{Z_q, S_i\}\big) \cup \{Z_q \cup S'\} \colon S' \subset S_i  \Big\}.
    \]
    Then we define $\parrep(p) = P_r$.
    Given a set $S\subseteq \Cp$, we define 
    \[\parrep(S) = \bigdelta\limits_{p \in S} \parrep(p).\]
\end{definition}
\begin{observation}\label{obs:prep-size-bound}
    For a complete pattern $p \in \Cp$, it holds that each pattern of $\parrep(p)$ results from $p$ by removing all non-singleton sets different from $Z_p$, and adding some subset of $\inc(p)$ to $Z_p$. Hence, it holds that $|\parrep(p)|\leq 2^{|\inc(p)|}$.
\end{observation}

\section{The Algorithm} \label{sec:algo}

In this section we present the main result of this paper, namely a dynamic programming algorithm over $\syntaxtree$ solving the \Coctp problem. We start by defining the family of states $\States$ that we use to index the vectors computed by the recursive formulas of the dynamic programming algorithm. 

\subsection{Indices of the algorithm}

\subsubsection*{The family of CS-Patterns}

The family $\CSP$ is of a special interest in this paper, since we use it, similar to \cite{DBLP:journals/corr/abs-2307-14264/BojikianK23}, to count connectivity patterns of partial solutions over the graphs $G_x$ for all nodes $x \in \nodes$. Note that a $CS$-pattern $p$ is uniquely defined by its set of labels and its zero-set; Given $X, Y \subseteq U$ the set of labels and the zero-set of $p$ respectively, where $Y\subseteq X$, it holds that 
\[
    p = \big\{ \{u\} \colon u \in X \big\} \cup \big\{ \{0\} \cup Y \big\}.
\]

We define the partial ordering $\pleq$ over $\CSP$, where for two $CS$-patterns $p,q\in\CSP$ it holds that $p\pleq q$ if and only if $Z_p\subseteq Z_q$ and $\lbs(p)\subseteq \lbs(q)$. Clearly, $(\CSP, \pleq)$ defines a lattice, with $p\lor q = r$ where $\lbs(r) = \lbs(p)\cup \lbs(q)$, and $Z_r = Z_p \cup Z_q$. Hence, the join operation over this lattice corresponds to the union operation over $CS$-patterns $p\punion q$.

\subsubsection*{Colorings and graph bipartition}

We define the set $C_0 = \{\black, \white\}$, and call it the set of \emph{basic colors}, where $\black$ stands for black, and $\white$ stands for white.
We define the set of \emph{colors} $\colors=\{\noc, \black,\white, \bw\}$, where each element of $\colors$ corresponds to a different subset of $C_0$ in the natural way. We define the ordering $\sleq$ over $\colors$ given by inclusion over the corresponding subsets, namely, given by the two chains $\noc\sleq\black\sleq \bw$ and $\noc\sleq \white\sleq \bw$.
We denote by $\sjoin$ the join operation, and by $\smeet$ the meet operation over the underlying lattice. We call two colors $x,y\in \colors$ \emph{consistent} (denoted by $x\sim y$), if $x\smeet y = \noc$.

We call a mapping $c:U\rightarrow \colors$ a coloring. We denote the family of all colorings of $U$ by $\Coloring= \colors^{U}$. We define the ordering $\cleq$ over colorings, where for two colorings $c_1,c_2\in \Coloring$, it holds that $c_1\cleq c_2$ if it holds that $c_1(i)\sleq c_2(i)$ for all $i\in U$. Finally, we denote by $\cjoin$ the join operation over the underlying lattice, i.e.\ for $c_1, c_2 \in \Coloring$ and for all $i\in U$ it holds
\[c_1\cjoin c_2(i) = c_1(i) \sjoin c_2(i).\]

For $c\in \Coloring$ and $i,j \in [k]$ such that $i\neq j$, we define $c_{i\rightarrow j}:[k]\rightarrow \colors$ as
    \[c_{i\rightarrow j} (\ell) = 
    \begin{cases}
        \noc&\colon\ell=i,\\
        c(i)\sqcup c(j)&\colon\ell = j,\\
        c(\ell)&\colon\text{otherwise.}
    \end{cases}\]

\subsubsection*{Indices of the dynamic programming routine}

Let $\States = \CSP\times \Coloring$ be the family of all pairs of a $CS$-pattern and a coloring. We call a pair $(p, c)\in \CSP\times \Coloring$ a state, and $\States$ the family of all states. We index the dynamic programming tables by the elements of this family. While the pattern $p$ indicates the connectivity of a partial solution, the coloring $c$ indicates the \emph{basic colors} appearing in each label class in a witness of this partial solution, allowing to extend this partial solution correctly, preserving a proper 2-coloring of the rest of the graph.

We define the ordering $\ileq$ over $\States$ with $(p_1,c_1)\ileq (p_2,c_2)$ if and only if $p_1\pleq p_2$ and $c_1\cleq c_2$.
It holds that $\ileq$ is the product of two lattices, and hence, it holds that $(p_1,c_1)\lor(p_2,c_2)=(p,c)$, where $p=p_1\punion p_2$ and $c=c_1\cjoin c_2$.

\subsection{Description of the Algorithm}
Intuitively, actions (\cref{def:actions}) allow to build representations (in an existential sense) of different weights of all partial solutions by a dynamic programming scheme over $\syntaxtree$ using complete patterns only. However, since the number of complete patterns is at least the number of all partitions of all subsets of $U$, we seek to reduce the space of the states of the dynamic programming routine, by representing (in a count-preserving manner) these patterns using $CS$-patterns only (\cref{def:parrep}).

In this section, we define the vectors $T\nodeind\in \{0, 1\}^{\States}$ for $x\in\nodes$ and all values $\budget$ and $\weight$. These vectors constitute the dynamic programming tables of our algorithm. We show that these vectors can be computed in time $\ostar(12^k)$. 
In the following sections, we will prove the correctness of the algorithm in two stages. In the first stage we show that any pattern, corresponding to a partial solution at a node $x$, can be represented by a set of complete patterns that can be built using different actions at different nodes in $\wnodes_x$. We assign different weights to different actions at each node of $\syntaxtree$.
In the second stage, we show that the tables $T\nodeind$ count for each $CS$-pattern p, the number of representations of weight $w$ of solutions of size $\budget$ that are consistent with $p$. By appropriately choosing the weight function, we show that this suffices to solve the \Coctp problem with hight probability.

For all $x\in V(\syntaxtree)$ and all values of $\budget, \weight$ we define the vectors $T\nodeind \in \{0, 1\}^{\States}$ recursively over $\syntaxtree$ as follows:

\begin{itemize}
    \item Introduce node $\mu_x = i(v)$:
    For $X\in\{\noc,\black,\white\}$ we define $c_X:[k]\rightarrow\colors$ as
    \[
        c(\ell)= 
        \begin{cases}
            X&\colon \ell = i,\\
            \noc&\colon\text{otherwise.}
        \end{cases}
    \]
    Let $p = [0i]$ if $v = v_0$, and $p=[i]$ otherwise. Let $p_1 = \forget(p, i)$ and $p_2= \fix(p, i)$.
    We set all values $T\nodeind\ind{q, c}$ to zero, for all values of $q$ and $c$, and then we add one to each of the following entries:

    \begin{itemize}
        \item $T\ind{x, 1, \weightf(x,1)} \ind{(p_1, c_{\noc})}$,
        \item $T\ind{x, 1, \weightf(x,2)} \ind{(p_2, c_{\noc})}$,
        \item $T\ind{x, 0, \weightf(x, 3)}\ind{([0], c_{\black})}$,
        \item $T\ind{x, 0, \weightf(x, 4)}\ind{([0],c_{\white})}$.
    \end{itemize}

    \item Relabel node $\mu_x = \rho_{i\rightarrow j}(\mu_{x'})$:
    For each state $(p, c)\in \States$ we define
    \[
    T\nodeind\ind{(p,c)} = \sum\limits_{\substack{(p',c')\in\States,\\p'_{i\rightarrow j} = p \land c'_{i\rightarrow j} = c}} T\ind{x',\budget,\weight}\ind{(p',c')}.   
    \]

    \item Join node $\mu_x = \eta_{i,j} (\mu_{x'})$:
    For $(p,c)\in\States$, if $c(i) \not\sim c(j)$, we set 
    $T\nodeind\ind{(p,c)} =0$. Otherwise, we define 
    \[
    T\nodeind\ind{(p,c)} =
    \sum\limits_{\ell\in[4]}
    \sumstack{p'\in\CSP\\ p\in \parrep\big(\action(\patadd_{i,j}p', \ell)\big)}
    T\ind{x',\budget,\weight - \weightf(x, \ell)}\ind{(p',c)}.
    \]
    
    \item Union node $\mu_x = \mu_{x_1} \union \mu_{x_2}$:
    We define 
    \[
        T\nodeind\ind{(p,c)} =
        \sumstack{\budget_1+\budget_2=\budget\\\weight_1+\weight_2=\weight}
        \Big(\sumstack{p_1\punion p_2 = p\\c_1\cjoin c_2=c}
        T\leftind\ind{(p_1,c_1)} \cdot T\rightind\ind{(p_2,c_2)}\Big).
    \]
\end{itemize}

In the following, we show how to compute these tables in time $\ostar(12^k)$.
We assume that lookups and updates of a single entry of any of these vectors can be done in time logarithmic in the number of indices, and hence, in polynomial time in $n$.

Although computing the tables $T\nodeind$ for a union node $x$ in the naive way requires time polynomial in ${(12^k)}^2 = 144^k$, we show that one can apply fast convolution methods to compute these tables more efficiently.

\begin{theorem}[\cite{DBLP:journals/talg/BjorklundHKKNP16}]
    Let $(\mathcal{L},\preceq)$ be a finite lattice given in join-representation and $A,B\colon\mathcal{L}\rightarrow \mathbb{F}$ be two tables, where $\mathbb{F}$ is some field. The $\lor$-product $A\otimes_{\mathcal{L}}B$ can be computed in $O(|\mathcal{L}||\mathcal{L}_{\lor}|)$ field operations and calls to algorithm $\mathcal{A}_{\mathcal{L}}$ and further time $O(|\mathcal{L}||\mathcal{L}_{\lor}|^2)$.
\end{theorem}

In this paper, we are mainly interested in the following corollary from \cite{DBLP:journals/corr/abs-2302-03627/HegerfeldK23}:

\begin{corollary}[{\cite[Corollary A.10]{DBLP:journals/corr/abs-2302-03627/HegerfeldK23}}]\label{cor:falko-lattice}
    Let $(\mathcal{L}, \preceq)$ be a finite lattice given in the join-representation and $k$ be a natural number. Given two tables $A,B\colon\mathcal{L}^k\rightarrow \mathbb{Z}_2$, the $\lor$-product $A\otimes_{\mathcal{L}^k}B$ in $\mathcal{L}^k$ can be computed in time $O(k^2|\mathcal{L}|^{k+2})$ and $O(k|\mathcal{L}|^{k+1})$ calls to the algorithm $\mathcal{A}_{\mathcal{L}}$.
\end{corollary}

\begin{definition}
    Let $\lat_0=[3]\times [4]$. We define the lattice $(\lat_0,\preceq_0)$ with $(x_1,y_1)\preceq_0 (x_2, y_2)$ if and only if $x_1\leq x_2$, and one of the following is true: $y_2 = 4$, $y_1 = 1$ or $y_1 = y_2$. Let $\lat^* = \lat_0^k$ together with the ordering $\preceq$ be the $k$-th power of this lattice.
\end{definition}

\begin{definition}\label{lem:rho-iso}
    We define the bijective mapping $\rho\colon \States\rightarrow \lat^*$, as $\rho(p,c) = ((x_1,y_1),\dots (x_k,y_k))$, where for $i\in[k]$ it holds
    \[
        x_i = \begin{cases}
            3& \colon i \in Z_p,\\
            2& \colon i \in \lbs(p)\setminus Z_p,\\
            1& \colon \text{otherwise,}
        \end{cases}
        \hspace{1cm}\text{and}\hspace{1cm}
        y_i = \begin{cases}
            1& \colon c(i) = \noc,\\
            2& \colon c(i) = \black,\\
            3& \colon c(i) = \white,\\
            4& \colon c(i) = \bw.
        \end{cases}
    \]
    Clearly $\rho$ and its inverse can be computed in polynomial time in $k$.
\end{definition}

\begin{lemma}\label{lem:rho-iso-power}
    The mapping $\rho$ defines an isomorphism between $\States$ and $\lat^*$.
\end{lemma}

\begin{proof}
    Let $(p_1,c_1),(p_2,c_2)\in \States$. We show that $(p_1,c_1)\ileq (p_2,c_2)$ if and only if $\rho(p_1,c_1)\preceq \rho(p_2,c_2)$.
    Let $\rho(p_1,c_1) = (x^1_1,y^1_1, \dots, x^1_k, y^1_k)$ and $\rho(p_2,c_2) = (x^2_1,y^2_1,\dots,x^2_k, y^2_k)$.
    It holds that $p_1\pleq p_2$ if and only if $Z_{p_1} \subseteq Z_{p_2}$ and $\lbs(p_1)\subseteq \lbs(p_2)$. This is the case if and only if for all $i \in [k]$ it holds that $x^1_i = 3$ implies that $x^2_i = 3$, and $x^1_i = 2$ implies that $x^2_i \geq 2$, which holds if and only if $x^1_i \leq x^2_i$.

    On the other hand, it holds that $c_1\cleq c_2$ if and only if $c_1(i)\sleq c_2(i)$ for all $i\in[k]$, which is the case if and only if $c_2(i) = \bw$, $c_1(i)=\noc$ or $c_1(i) = c_2(i)$, which is the case if and only if $y^2_i = 4$, $y^1_i = 1$, or $y^1_i = y^2_i$ for all $i \in [k]$. Hence, $(p_1,c_1)\ileq (p_2,c_2)$ if and only if $(x^1_i,y^1_i)\preceq_0 (x^2_i, y^2_i)$ for all $i\in[k]$, which is the case if and only if $\rho(p_1, y_1)\preceq \rho(p_2,y_2)$.
\end{proof}

\begin{corollary}\label{cor:lor-prod-states-time}
    The $\lor$-product over $\States$ can be computed in time $\ostar(12^k)$.
\end{corollary}

\begin{proof}
    By \cref{lem:rho-iso-power} it holds that $\rho$ is an isomorphism between $\States$ and $\lat^*$ such that both $\rho$ and $\rho^{-1}$ can be computed in polynomial time in $k$. By \cref{cor:falko-lattice}, the $\lor$-product over $\lat^*$ can be  computed in time $\ostar(12^k)$, since  $|\lat_0| = 12$ and $\lat^* = (\lat_0)^k$. Hence, by \cref{lem:v-prod-iso}, the $\lor$-product over $\States$ can be computed in time $\ostar(12^k)$.
\end{proof}

\begin{corollary}\label{cor:union-node-time}
    Let $x$ be a union node of $\syntaxtree$, with $\mu_x = \mu_{x_1} \union \mu_{x_2}$. Given $T\ind{x_1,\budget, \weight}$ and $T\ind{x_2,\budget, \weight}$ for all values of $\budget$ and $\weight$, then $T\nodeind$ for all values of $\budget$ and $\weight$ can be computed in time $\ostar(12^k)$.
\end{corollary}

\begin{proof}
    It holds that 
    \[T\nodeind\ind{(p,c)} =
    \sumstack{\budget_1+\budget_2=\budget\\\weight_1+\weight_2=\weight}
    \Big(\sumstack{p_1\punion p_2 = p\\c_1\cjoin c_2=c}
    T\leftind\ind{(p_1,c_1)} \cdot T\rightind\ind{(p_2,c_2)}\Big).\]
    For a fixed tuple $(\budget_1,\budget_2,\weight_1,\weight_2)$, we define the tables $T\ind{x,(\budget_1,\budget_2,\weight_1,\weight_2)} \in \{0,1\}^{\States}$ by 
    \[T\ind{x, (\budget_1,\budget_2,\weight_1,\weight_2)}\ind{(p,c)} = \sumstack{p_1\punion p_2 = p\\c_1\cjoin c_2=c} T\leftind\ind{(p_1,c_1)} \cdot T\rightind\ind{(p_2,c_2)},\]
    i.e., the table $T\ind{x, (\budget_1,\budget_2,\weight_1,\weight_2)}$ is the $\lor$-product of $T\leftind$ and $T\rightind$ over $\States$. Hence, by \cref{cor:lor-prod-states-time}, it can be computed in time $\ostar(12^k)$. On the other hand, it holds that
    \[
        T\nodeind\ind{(p,c)} =
    \sumstack{\budget_1+\budget_2=\budget\\\weight_1+\weight_2=\weight}
    T\ind{x, (\budget_1,\budget_2,\weight_1,\weight_2)}.
    \]
    Since we iterate over at most polynomially many different tuples to compute all tables $T\nodeind$, the lemma follows.
\end{proof}

\begin{lemma}
    All tables $T\nodeind$ for all values of $x,\budget, \weight$ can be computed in time $\ostar(12^k)$.
\end{lemma}

\begin{proof}
    This is clearly the case for an introduce node. For a relabel node,
    for all values of $\budget$ and $\weight$, we initialize $T\nodeind$ to $\overline 0$ and iterate over all states in $\States$. For each such pair $(p,c)$, we add $T\subnodeind\ind{(p,c)}$ to $T\nodeind\ind{(p_{i\rightarrow j}, c_{i\rightarrow j})}$.

    For a join node, for all values of $\budget$ and $\weight$ we initialize $T\nodeind$ to $\overline 0$. Then we iterate again over all values of $\budget$ and $\weight$, and over all pairs $(p,c) \in \States$. For each such pair, if $T\ind{x',\budget,\weight}\ind{(p,c)}= 0$ or  $c(i)$ and $c(j)$ are not consistent, we skip this pair. Otherwise, let $p' = \patadd_{i,j}p$. For each $\ell\in [4]$, let $p_{\ell} = \action(p', \ell)$ and $P_{\ell} = \parrep(p_{\ell})$.
    Then we add one to each entry $T\ind{x, \budget, \weight + \actionf(x, \ell)}\ind{(q, c)}$ for each $q \in P_{\ell}$.

    Since $\{i,j\}$ is the only non-singleton set different from $Z_{p_{\ell}}$ that can appear in $p_{\ell}$, it holds that $\inc(p_{\ell}) \subseteq \{i, j\}$. Hence, it holds by \cref{obs:prep-size-bound} that $|P_{\ell}|\leq 4$. In total, the algorithm performs at most 16 binary addition operations for each fixed values of $\budget$ and $\weight$, and for each pair $(p, c) \in \States$.
    
    Finally, for a union node, the running time follows from \cref{cor:union-node-time}.
\end{proof}

In the following sections, we show that, with high probability, there exists a value $\weight \in [\W \cdot V(\syntaxtree)]$ and a coloring $c\colon [k]\rightarrow \colors$ with $T\rootind\ind{([0], c)} = 1$, if and only if there exists a connected odd cycle transversal of size $b$ in~$G$, proving the correctness of our algorithm.

\section{Solution representation}\label{sec:acseq}

\subsection{Action-sequences}

\begin{definition}\label{def:ac-seq}
    Let $x\in \nodes$. An \emph{action-sequence} at $x$ is a mapping $\tau:\wnodes_x\rightarrow [4]$. We define the weight of $\tau$ as
    \[\weightf(\tau) = \sum\limits_{y\in \wnodes_x}\weightf(y, \tau(x)),\]
    and the cost of $\tau$ as 
    \[\budgetf(\tau) = \big|\big\{y \in \cnodes_x \colon \tau(y)\in\{1,2\}\big\}\big|.\]
    
    Each node $y \in V(\syntaxtree_x)$ induces an \emph{action-subsequence} $\tau_y$ of $\tau$ at $y$ (denoted by $\tau_y$) given by the restriction of $\tau$ to $\wnodes_y$.
    Each action-sequence $\tau$ at a node $x$ generates a pattern $\acpat{\tau}{x}\in \Cp$ and a coloring $c^{\tau}_x \in \colors^{[k]}$ (defined next). We say that $\tau$ \emph{generates} the pair $(\acpat{\tau}{x}, c^{\tau}_x)$.
    For $y\in \nodes_x$, we denote $\acpat{\tau_y}{y}$ by $\acpat{\tau}{y}$, and $c^{\tau_y}_y$ by $c^{\tau}_y$. 
    We call an action-sequence $\tau$ at a node $x$ \emph{valid}, if and only if for each join node $y\in \nodes_x$, where $\mu_y = \eta_{i,j} (\mu_{y'})$, it holds that that $c^{\tau}_{y'}(i)$ and $c^{\tau}_{y'}(j)$ are consistent.
    
    We define $(\acpat{\tau}{y}, c^{\tau}_y)$ recursively as follows:
\begin{itemize}
    \item Introduce node $\mu_x = i(v)$:
    For $X\in\{\noc,\black,\white\}$ we define $c^i_X:[k]\rightarrow\colors$ as
    \[
        c^i_X(\ell)= 
        \begin{cases}
            X&\colon \ell = i,\\
            \noc&\colon\text{otherwise.}
        \end{cases}
    \]
    Let $p = [0i]$ if $v = v_0$, and $p=[0,i]$ otherwise. We distinguish different values of $\tau(x)$:

    \begin{itemize}
        \item if $\tau(x) = 1$, we define $\acpat{\tau}{x} = \forget(p, i)$ and $c^{\tau}_x = c^i_{\noc}$,
        \item if $\tau(x) = 2$, we define $\acpat{\tau}{x} = \fix(p, i)$ and $c^{\tau}_x = c^i_{\noc}$,
        \item if $\tau(x) = 3$, we define $\acpat{\tau}{x} = [0]$ and $c^{\tau}_x = c^i_{\black}$,
        \item if $\tau(x) = 4$, we define $\acpat{\tau}{x} = [0]$ and $c^{\tau}_x = c^i_{\white}$.
    \end{itemize}

    \item Relabel node $\mu_x = \rho_{i\rightarrow j}(\mu_{x'})$: we define
    \[ \acpat{\tau}{x} = (\acpat{\tau}{x'})_{i\rightarrow j},\quad \text{and }\quad c^{\tau}_x = (c^{\tau}_{x'})_{i\rightarrow j}.\]

    \item Join node $\mu_x = \eta_{i,j} (\mu_{x'})$: Let $p = \patadd_{i,j} \acpat{\tau}{x'}$. We define
    \[ \acpat{\tau}{x} = \action(p, \tau(x)),\quad \text{and }\quad c^{\tau}_x = c^{\tau}_{x'}.\]
    
    \item Union node $\mu_x = \mu_{x_1} \union \mu_{x_2}$: we define
    \[ \acpat{\tau}{x} = \acpat{\tau}{x_1} \punion \acpat{\tau}{x_2}, \quad \text{and } \quad c^{\tau}_x = c^{\tau}_{x_1} \cjoin c^{\tau}_{x_2}.\]
\end{itemize}
\end{definition}

\subsection{Solution representation through action-sequences.}

\begin{definition}\label{def:representation}
    Given two families $S, R \subseteq \Pat$,
    we say that $R$ \emph{represents} $S$, if for each $q \in \Pat$ the following holds: there exists a pattern $p \in S$ such that $p\sim q$ if and only if there exists a pattern $p'\in R$ such that $p'\sim q$.
    We say that the family $S$ \emph{represents} a pattern $p\in\Pat$, if $S$ represents $\{p\}$.
\end{definition}

\begin{observation}\label{obs:rep-closed-under-set-union}
    Given four families $S_1, S_2, R_1, R_2 \subseteq \Pat$, where $R_1$ represents $S_1$, and $R_2$ represents $S_2$, it holds that $R_1\cup R_2$ represents $S_1\cup S_2$.
\end{observation}

\begin{definition}\label{def:rep-family-Rp}
	Given a pattern $p\in\Pat$ with $|\inc(p)|\leq 2$, we define the set $R_p$ as follows: if $p\in \Cp$, then we set $R_p = \{p\}$. Otherwise, we set $R_p = \{\action(p, i)\colon i \in [r_p]\}$, where $r_p = 2$ if $|\inc(p)| =1$, and $r_p = 4$ if $|\inc(p)| = 2$.
\end{definition}

\begin{lemma}[{\cite[Lemma 6.9]{DBLP:journals/corr/abs-2307-14264/BojikianK23}}] \label{lem:actions-represent}
	Let $p\in \Pat$ with $|\inc(p)|\leq 2$, it holds that $R_p$ represents $p$.
\end{lemma}

\begin{definition} \label{def:op-pres-rep}
    Given a $k$-ary operation over patterns $\op\colon \Pat^k \rightarrow \Pat$, we say that $\op$ preserves representation, if
    for all sets $S_1,\dots, S_k, T_1,\dots T_k \subseteq \Pat$ such that $T_i$ represents $S_i$ for all $i\in[k]$ it holds that $\op(S_1,\dots, S_k)$ represents $\op(T_1,\dots, T_k)$.
\end{definition}

\begin{lemma}[{\cite[Lemma 6.4 and Lemma 6.5]{DBLP:journals/corr/abs-2307-14264/BojikianK23}}] \label{lem:ops-pres-rep}
    The operations join, relabel and union, as defined in \cref{def:patops}, as well as the operation $\patadd_{i,j}$ defined in \cref{def:patadd} for all $i,j\in[k]$ with $i\neq j$ preserve representation.
\end{lemma}

\begin{definition}\label{def:sol-pattern-n-coloring}
    Given a labeled graph $G = (V, E)$ and a set of vertices $S$, we denote with $\solpat{S}{G}$ the pattern defined as follows: let $C_0 = \emptyset$ if $v_0 \notin S$, or let $C_0$ be the connected component of $G[S]$ containing $v_0$ otherwise. Let $C_1, \dots C_r$ be the connected components of $G[S]$ not containing $v_0$. Then we define
    \[ \solpat{S}{G} = \big\{\lab(C_i)\colon i\in [r]\big\} \cup \big\{\{0\} \cup C_0\big\},\]
    where we add a set for each connected component containing the labels that appear in this component. We add the label $0$ to the set corresponding to the component containing $v_0$, if $v_0 \in S$, or as a singleton otherwise.
    
    Moreover, given a mapping $g\colon S \rightarrow \{\black, \white\}$, we denote by $c = \solcol{g}{G}$ the coloring $c\colon [k] \rightarrow \colors$ defined as follows:
    \[
    c(\ell) = 
    \begin{cases}
        \noc   & \colon g\big(\lab_{G[S]}^{-1}(\ell)\big) = \emptyset, \\ 
        \black & \colon g\big(\lab_{G[S]}^{-1}(\ell)\big) = \{\black\}, \\
        \white & \colon g\big(\lab_{G[S]}^{-1}(\ell)\big) = \{\white\}, \\
        \bw    & \colon g\big(\lab_{G[S]}^{-1}(\ell)\big) = \{\black, \white\}. \\
    \end{cases}    
    \]
    
    We denote by $\solpat{S}{x}$ the pattern $\solpat{S_x}{G_x}$, and by $\solcol{g}{x}$ the coloring $\solcol{g_x}{G_x}$ for short.
\end{definition}

In the following, we show that one can represent any partial solution by action-sequences. In order to achieve this we introduce the notion of \emph{solution-sequences}. Similar to an action-sequence, a solution-sequence is defined by operations over $\syntaxtree$. However, a solution-sequence only indicates actions, that specify which vertices are included in a partial solution, and hence, it corresponds to a whole solution, instead of a representing pattern thereof.

\begin{definition}\label{def:sol-seq}
    Let $x\in \nodes$. A \emph{solution-sequence} at $x$ is a mapping $\pi:\cnodes_x\rightarrow \{0,3,4\}$. We define the cost of $\pi$ as
    \[\budgetf(\pi) = |\{y \in \cnodes_x \colon \pi(y) = 0\}|.\]
    
    Each node $y \in V(\syntaxtree_x)$ induces a \emph{solution-subsequence} $\pi_y$ of $\pi$ at $y$ given by the restriction of $\pi$ to $\cnodes_y$.
    Each solution-sequence $\pi$ at a node $x$ generates a pattern $\cpat{\pi}{x}\in \Pat$ and a coloring $\ccol{\pi}{x} \in \colors^{[k]}$ (defined next). We say that $\pi$ \emph{generates} the pair $(\cpat{\pi}{x}, \ccol{\pi}{x})$.
    For $y\in V(\syntaxtree_x)$, we denote $\cpat{\pi_y}{y}$ by $\cpat{\pi}{y}$, and $c^{\pi_y}_y$ by $\ccol{\pi}y$. 
    We call a solution-sequence $\tau$ at a node $x$ \emph{valid}, if and only if for each join node $y\in \nodes_x$, where $\mu_y = \eta_{i,j} (\mu_{y'})$, it holds that $c^{\tau}_{y'}(i)$ and $c^{\tau}_{y'}(j)$ are consistent.
    
    We define $(\cpat{\tau}{y}, c^{\tau}_y)$ recursively as follows:
\begin{itemize}
    \item Introduce node $\mu_x = i(v)$:
    For $X\in\{\noc,\black,\white\}$ we define $c_X:[k]\rightarrow\colors$ as
    \[
        c^i_X(\ell)= 
        \begin{cases}
            X&\colon \ell = i,\\
            \noc&\colon\text{otherwise.}
        \end{cases}
    \]
    Let $p = [0i]$ if $v = v_0$, and $p=[0,i]$ otherwise. We distinguish different values of $\pi(x)$:
    \begin{itemize}
        \item if $\pi(x) = 0$, we define $\cpat{\pi}{x} = p$ and $\ccol{\pi}{x} = c^i_{\noc}$,
        \item if $\pi(x) = 3$, we define $\cpat{\pi}{x} = [0]$ and $\ccol{\pi}{x} = c^i_{\black}$,
        \item if $\pi(x) = 4$, we define $\cpat{\pi}{x} = [0]$ and $\ccol{\pi}{x} = c^i_{\white}$.
    \end{itemize}

    \item Relabel node $\mu_x = \rho_{i\rightarrow j}(\mu_{x'})$: we define
    \[ \cpat{\pi}{x} = (\cpat{\pi}{x'})_{i\rightarrow j},\quad \text{and }\quad \ccol{\pi}{x} = (\ccol{\pi}{x'})_{i\rightarrow j}.\]

    \item Join node $\mu_x = \eta_{i,j} (\mu_{x'})$: We define
    \[ \cpat{\pi}{x} = \patadd_{i,j} \cpat{\pi}{x'},\quad \text{and }\quad \ccol{\pi}{x} = \ccol{\pi}{x'}.\]
    
    \item Union node $\mu_x = \mu_{x_1} \union \mu_{x_2}$: we define
    \[ \cpat{\pi}{x} = \cpat{\pi}{x_1} \punion \cpat{\pi}{x_2}, \quad \text{and } \quad \ccol{\pi}{x} = \ccol{\pi}{x_1} \cjoin \ccol{\pi}{x_2}.\]
\end{itemize}
\end{definition}

\begin{definition}\label{def:sol-seq-from-sol}
    Let $x \in \nodes$, and $S \subseteq V_x$. Let $g : V_x \setminus S \rightarrow \{\black, \white\}$ be some mapping. Let $\pi = \pi^{S,g}_x:\cnodes_x \rightarrow \{0,3,4\}$ be the solution-sequence at $x$ defined as follows: for $y\in \cnodes_x$, let $v \in V_x$ be the vertex introduced at $y$, i.e.\ $\mu_y = i(v)$ for some $i\in [k]$. We define
    \[
    \pi(y) =
    \begin{cases}
    0&\colon v \in S,\\
    3&\colon v \notin S \land g(v) = \black,\\
    4&\colon v \notin S \land g(v) = \white.\\
    \end{cases}
    \]
\end{definition}

\begin{lemma}\label{lem:sol-seq-equiv-solution}
    Let $x \in \nodes$. For $S \subseteq V_x$ and some mapping $g : V_x \setminus S \rightarrow \{\black, \white\}$, let $\pi = \pi^{S,g}_x$.
    Then all the following holds:
    \begin{enumerate}
        \item $\budgetf(\pi) = |S|$,\label{it:sol-seq-quiv-sol-same-weight}
        \item $\solpat{S}{x}=\cpat{\pi}{x}$,
        \item $\solcol{g}{x} = \ccol{\pi}{x}$,
        \item and $g$ is a proper 2-coloring of $G_x[V_x\setminus S]$ if and only if $\pi$ is valid.
    \end{enumerate}
\end{lemma}

\begin{proof}
    First, we prove \cref{it:sol-seq-quiv-sol-same-weight}. It holds that 
    \[
        \budgetf(\pi) = |\{y \in \cnodes_x \colon \pi(y) = 0\}| = |v \in S|,
    \]
    where the first equality holds from \cref{def:sol-seq}, and the second from \cref{def:sol-seq-from-sol}.
    Now we prove the other points for all $y\in \nodes_x$ by induction over $\syntaxtree_x$. 
    We distinguish the different types of nodes $y$.
    \begin{itemize}
        \item Introduce node $\mu_y = i(v)$ (base case):
        If $v\in S$, then it holds that $\pi(y) = 0$. Hence, it holds for $p = [0i]$ if $v = v_0$, or $p = [i]$ otherwise, that $\solpat{S}{y} = \cpat{\pi}{y} = p$. It also holds that $\solcol{g}{y}(j) = \ccol{\pi}y(j) = \noc$ for $j\in[k]$.
    
        Now we assume that $v\notin S$. It follows that $\pi(y) \neq 0$. Hence, it holds that $\solpat{S}{y} = \cpat{\pi}{y} = [0]$. It holds that $\pi(y) = 3$ if $g(v) = \black$ and $\pi(y) = 4$ otherwise. In both cases, it holds that $\solcol{g}{y}(i) = \ccol{\pi}y(i) = g(v)$, and $\solcol{g}{y}(j) = \ccol{\pi}y(j) = \noc$ for all $j\in[k] \setminus \{i\}$.

        The mapping $g$ is clearly a proper 2-coloring. Since there exist no join node in the subtree $\syntaxtree_y$, it holds that $\pi_y$ is a valid solution-subsequence at $y$.
    \end{itemize}
        Note that for a node $y\in \nodes_x$, if there exists a child $y'$ of $y$ where $g_{y'}$ is not a proper 2-coloring of $G_{y'}\setminus S$, then it holds by induction hypothesis, that $\pi_{y'}$ is not a valid solution-subsequence. Hence, it follows from the definition of a valid solution-sequence, that $\pi_y$ is not valid as well. Since $G_{y'}$ is a subgraph of $G_y$, and $g_y$ is an extension of $g_{y'}$, it holds in this case that $g_y$ is not a proper 2-coloring of $G_y \setminus S$ as well.
        Hence, we can assume that $g_{y'}$ is a proper 2-coloring of $G_{y'}\setminus S$, and that $\pi_{y'}$ is a valid solution-subsequence for all children $y'$ of $y$ in $\syntaxtree$.
    \begin{itemize}
        \item Relabel node $\mu_y = \rho_{i\rightarrow j}(\mu_{y'})$:
        Let $p = \solpat{S}{y'}$, and $c = \solcol{g}{y'}$. It holds by induction hypothesis that $\cpat{\pi}{y'} = p$ and $c = \ccol{\pi}{y'}$.
        The graph $G_y$ results from $G_{y'}$ by changing the label of each vertex labeled $i$ in $G_{y'}$ to the label $j$. Hence, it holds clearly that 
        \[
            \solpat{S}{y} = \big(\solpat{S}{y'}\big)_{i\rightarrow j} = p_{i\rightarrow j} = \cpat{\pi}{y}
        \]
        and
        \[
            \solcol{g}{y} = \big(\solcol{g}{y'}\big)_{i\rightarrow j} = c_{i\rightarrow j} = \ccol{\pi}{y}.
        \]

        Since we assume that $\pi_{y'}$ is a valid solution-subsequence at $y'$, it holds that $\pi_y$ is a valid solution-subsequence at $y$ as well. Also since $g_y = g_{y'}$ is a proper 2-coloring of $G_{y'}\setminus S$, and since $G_y$ contains the same sets of vertices and edges as $G_{y'}$ (but possibly a different labeling), it holds that $g_y$ is a proper 2-coloring of $G_y\setminus S$ as well.
    
        \item Join node $\mu_y = \eta_{i,j} (\mu_{y'})$:
        Let $p = \solpat{S}{y'}$, and $c = \solcol{g}{y'}$. It holds by induction hypothesis that $\cpat{\pi}{y'} = p$ and $c = \ccol{\pi}{y'}$.
        The graph $G_y$ results from $G_{y'}$ by adding all edges between vertices of label $i$ and vertices of label $j$, turning them into one connected component if vertices of both labels exist in $G_y$, or by keeping $G_{y'}$ as it is otherwise. Hence, it holds clearly that 
        \[
            \solpat{S}{y} = \patadd_{i,j}\big(\solpat{S}{y'}\big) = \patadd_{i,j}p = \cpat{\pi}{y}
        \]
        and
        \[
            \solcol{g}{y} = \solcol{g}{y'} = c = \ccol{\pi}{y}.
        \]

        Since we assume that $\pi_{y'}$ is a valid solution-subsequence at $y'$, it holds that $\pi_y$ is a valid solution subsequence at $y$ if and only if $\solcol{g}{y}(i) = c(i)$ and $\solcol{g}{y}(j) = c(j)$ are consistent, which is the case if and only if for each edge $\{u, v\} \in E_y$ with $\lab_{y}(u) = i$ and $\lab_y(v) = j$ it holds that $g(u)\neq g(v)$. However, this is the case if and only if $g_y$ is a proper 2-coloring of $G_y\setminus S$, since only such edges are added to $G_{y'}$, and since $g_{y'}$ is a proper 2-coloring of $G_{y'}\setminus S$.
        
        \item Union node $\mu_y = \mu_{y_1} \union \mu_{y_2}$: 
        Let $p_1 = \solpat{S}{y_1}, p_2 = \solpat{S}{y_2}, c_1 = \solcol{g}{y_1}$ and $c_2 = \solcol{g}{y_2}$. It holds by induction hypothesis that 
        \begin{itemize}
            \item $\cpat{\pi}{y_1} = \solpat{S}{y_1} = p_1$,
            \item $\cpat{\pi}{y_2} = \solpat{S}{y_2} = p_2$,
            \item $\ccol{\pi}{y_1} = \solcol{g}{y_1} = c_1$,
            \item and $\ccol{\pi}{y_2} = \solcol{g}{y_2} = c_2$.
        \end{itemize}
        Let $p = p_1 \punion p_2$, and $c = c_1 \cjoin c_2$. It follows from \cref{def:sol-seq}, that $\cpat{\pi}{y} = p$ and 
        $\ccol{\pi}{y} = c$.
        The graph $G_y$ results from the disjoint union of $G_{y_1}$ and $G_{y_2}$. Hence, the connected components of $G_y$ consist of the union of the connected components of both graphs $G_{y_1}$ and $G_{y_2}$. Since $v_0$ can belong to at most one of $V_{y_1}$ and $V_{y_2}$, it holds as well that the zero-set of at least one of $p_1$ and $p_2$ is a singleton. If follows that $\solpat{S}{y} = p_1\punion p_2 = p$. 
        
        Moreover, it holds for $X\in \{\black, \white\}$ that $X \sleq \solcol{g}{y}(i)$ if and only if there exists $v\in V_y$ with $\lab_y(v) = i$ and $g(v) = X$. Since it holds that $v \in V_{y_1}\dot\cup V_{y_2}$, this is the case, if and only if there exists $\ell \in [2]$ with $v \in V_{y_{\ell}}$ with $\lab_{y_{\ell}}(v) = \lab_y(v) = i$, which is the case if and only if $X \sleq \solcol{g}{y_{\ell}}(i)$. It follows that $\solcol{g}{y}(i) = \solcol{g}{y_1}(i)\sjoin \solcol{g}{y_2}(i)$ for all $i \in [k]$, and hence, $\solcol{g}{y} = \solcol{g}{y_1} \cjoin \solcol{g}{y_2} = c$.
        
        Since we assume that both $\pi_{y_1}$ and $\pi_{y_2}$ are valid solution-subsequence at $y_1$ and $y_2$ respectively, it holds by definition that $\pi_y$ is a valid solution-subsequence at $y$ as well. Also since $g_y = g_{y_1}\dot\cup g_{y_2}$, both proper 2-colorings of the corresponding subgraphs, and since $G_y$ is the disjoint union of $G_{y_1}$ and $G_{y_2}$, it holds that $g_y$ is a proper 2-coloring of $G_y\setminus S$ as well.
    \end{itemize}
\end{proof}

\begin{lemma}\label{lem:action-seq-rep-sol}
    Let $x\in \nodes$, $q\in \Pat$ and $c \in \Coloring$. There exists a set of vertices $S\subseteq V_x$ of size $\budget$ and a mapping $g:V_x\setminus S \rightarrow \{\black, \white\}$, such that $\solpat{S}{x} \sim q$, and $c=\solcol{g}{x}$, if and only if there exists a complete pattern $p'\in \Cp$ with $p'\sim q$, a weight $\weight$ and an action-sequence $\tau$ of cost $\budget$ and weight $\weight$ at $x$ generating the pair $(p', c)$. Moreover, $g$ is a proper 2-coloring of $G_x\setminus S$ if and only if $\tau$ is valid.
\end{lemma}

\begin{proof}
    Let $\pi:\cnodes_x\rightarrow\{0,3,4\}$ be an arbitrary solution-sequence at $x$, and let $(p, c)$ be the pair generated by $\pi$. Let $\Pi = \Pi_x(\pi)$ be the family of all action-sequences $\tau$ at $x$, such that for each introduce node $y\in\cnodes_x$, it holds that $\tau(x) = \pi(x)$ if $\tau(x) \in \{3,4\}$, or $\tau(x) \in \{1,2\}$ if $\pi(x) = 0$. We define the family of patterns $P = P_x(\pi) = \{\acpat{\tau}{x}\colon \tau \in \Pi\}$.
    By induction over $\syntaxtree$, we will prove that all following points hold:
    \begin{enumerate}
        \item The family $P_x(\pi)$ represents $p$.\label{it:ac-seq-to-sol-sec-represent}
        \item It holds for all $\tau\in \Pi$ that $c^{\tau}_x = c$.
        \item It holds for all $\tau \in \Pi$ that $\tau$ is valid if and only if $\pi$ is.\label{it:ac-seq-to-sol-sec-valid-if-valid}
        \item It holds for all $\tau \in \Pi$ that $\budgetf(\tau)=\budgetf(\pi)$.\label{it:ac-seq-to-sol-seq-cost}
    \end{enumerate}

    Before we prove this claim, we show that the lemma follows from it. So we assume first that the claim holds. Let $S \subseteq V_x$ be a set of vertices, and $g : V_x \setminus S \rightarrow \{\black, \white\}$ be some mapping. Let $p = \solpat{S}{x}$, $c = \solcol{g}{x}$ and $\pi = \pi^{S,g}_x$. It holds by \cref{lem:sol-seq-equiv-solution} that $\cpat{\pi}{x} = p$, $\ccol{\pi}{x} = c$, $\budgetf(\pi)=|S|$, and that $g$ is a proper 2-coloring of $V_x \setminus S$ in $G_x$ if and only if $\pi$ is a valid solution-sequence at $x$.
    
    Hence, it holds for $q\in \Pat$ by the claim above (\cref{it:ac-seq-to-sol-sec-represent}), that $p\sim q$ if and only if there exists a pattern $p'\in P$ with $p' \sim q$. Hence, it holds that $p\sim q$ if and only if there exists an action-sequence $\tau \in \Pi_x(\pi)$ of cost $|S|$ at $x$ generating the pair $(p', c)$ for some pattern $p'$ with $p'\sim q$. Finally, it holds by \cref{it:ac-seq-to-sol-sec-valid-if-valid} that $\tau$ is a valid action-sequence at $x$ if and only if $g$ is a proper two-coloring of $G_x\setminus S$.

    We first prove \cref{it:ac-seq-to-sol-seq-cost}. It holds that 
    \begin{alignat*}{2}
        \budgetf(\tau) &= \big|\big\{y \in \cnodes_x \colon \tau(y)\in\{1,2\}\big\}\big|&&\\
        &= |\{y \in \cnodes_x \colon \pi(y) = 0\}|&& = \budgetf(\pi),
    \end{alignat*}
    where the first equality follows from \cref{def:ac-seq}, the second from the definition of $\Pi_x(\pi)$, and the third from \cref{def:sol-seq}.
    Now we prove the other points by induction over $\syntaxtree$. Let $\pi$ be some solution-sequence at $x$ of cost $\budget$, and let $(p, c)$ be the pair generated by $\pi$. Let  $\Pi = \Pi_x(\pi)$, and $P = P_{x}(\pi)$.
    \begin{itemize}
        \item Introduce node $\mu_x = i(v)$: If $\pi(x) = 0$, it holds for $p = \cpat{\pi}{x}$, that $P = R_p$ (\cref{def:rep-family-Rp}). It follows from \cref{lem:actions-represent} that $P$ represents $p$. It holds for $\Pi = \{\tau_1, \tau_2\}$ that all sequences $\pi, \tau_1, \tau_2$ generate the coloring $c$ with $c(j) = \noc$ for all $j\in[k]$. On the other hand, if $\pi(x)\in\{3,4\}$, it holds that $\cpat{\pi}{x} = [0]$, and $P = \{[0]\}$, so $P$ trivially represents $p$. In this case $\Pi$ is a singleton $\tau$ with $\tau(x) = \pi(x)$. Hence, it holds that $\accol{\tau}{x} = \ccol{\pi}{x} = c$, where $c(j) = \noc$ for all $j\in[k]\setminus \{i\}$, and $c(i) = \black$ if $\pi(x) =3$ or $c(i) = \white$ otherwise. Finally, it holds that all mentioned sequences are valid ones.
            
        \item Relabel node $\mu_x = \rho_{i\rightarrow j}(\mu_{x'})$: Let $\pi' = \pi_{x'}$, $p' = \cpat{\pi'}{x'}$ and $c' = \ccol{\pi'}{x'}$. Let $\Pi' = \Pi_{x'}(\pi')$, and $P' = P_{x'}(\pi')$. It holds by \cref{def:sol-seq} that $\pi' = \pi$, and hence, by the definition of $\Pi$, that $\Pi' = \Pi$. It follows that $\cpat{\pi}{x} = (\cpat{\pi_{x'}}{x'})_{i\rightarrow j}$ and $P = P'_{i\rightarrow j}$. By induction hypothesis, it holds that $P'$ represents $p'$, and hence, by \cref{lem:ops-pres-rep} it holds that $P$ represents $p$. Also by induction hypothesis, it holds for each $\tau' \in \Pi'$ that $\tau'$ generates the coloring $c'$ at $x'$ and that $\tau'$ is valid if and only if $\pi'$ is. Since $\Pi' = \Pi$, it holds that $\tau_{x'} = \tau$ is in $\Pi'$ for each $\tau \in \Pi$. Hence, it follows that $\tau$ generates the coloring ${(c')}_{i \rightarrow j} = c$, and $\tau$ is valid, if and only if $\tau_{x'}$ is valid. This is the case if and only if $\pi'$ is valid, which is the case if and only if $\pi$ is valid.

        \item Join node $\mu_x = \eta_{i,j} (\mu_{x'})$: Let $\pi' = \pi_{x'}$, $p' = \cpat{\pi'}{x'}$ and $c' = \ccol{\pi'}{x'} = c $. Let $\Pi' = \Pi_{x'}(\pi')$, and $P' = P_{x'}(\pi')$. It holds by \cref{def:sol-seq} that 
        $\pi' = \pi$, and hence, by the definition of $\Pi$, that 
        \[
          \Pi =  \big\{\tau \dot\cup \{x\mapsto i\} \colon \tau \in \Pi' \land i\in[4]\big\}.
        \]
        It holds that $\cpat{\pi}{x} = \patadd_{i,j}\cpat{\pi_{x'}}{x'}$. Let $\hat{P} = \patadd_{i,j}P'$.
        By induction hypothesis, it holds that $P'$ represents $p'$, and hence, by \cref{lem:ops-pres-rep} it holds that $\hat{P}$ represents $p$.
        Since it holds from \cref{def:rep-family-Rp} for $q\in \hat P$, that $R_q = \{\action(q, i)\colon i\in[4]\}$, it holds from the definition of $\acpat{\tau}{x}$ (\cref{def:ac-seq}) that
        \[
        P = \bigcup\limits_{\ell\in[4]}\action(\hat P, \ell) = \bigcup\limits_{q \in \hat P}R_q.
        \]
        Finally, it holds by \cref{lem:actions-represent} and \cref{obs:rep-closed-under-set-union}  that $P$ represents $\hat{P}$, and hence, by transitivity of representation, that $P$ represents $p$ as well.

        By induction hypothesis, it holds for each $\tau' \in \Pi'$ that $\tau'$ generates the coloring $c'$ at $x'$ and that $\tau'$ is valid if and only if $\pi'$ is.
        For each $\tau \in \Pi$, it holds that $\tau_{x'} \in \Pi'$ and that $\tau$ and $\tau_{x'}$ generate the same coloring. Hence, $\tau$ generates the coloring $c' = c$ as well. Finally, $\tau$ is valid, if and only if $\tau_{x'}$ is valid and $c(i)\sim c(j)$. This is the case if and only if $\pi'$ is valid and $c(i)\sim c(j)$, which is the case if and only if $\pi$ is valid.
        
        \item Union node $\mu_x = \mu_{x_1} \union \mu_{x_2}$: Let $\pi_1 = \pi_{x_1}$, $p_1 = \cpat{\pi_1}{x_1}$, $c_1 = \ccol{\pi_1}{x_1}$, $\Pi_1 = \Pi_{x_1}(\pi_1)$, and $P_1 = P_{x_1}(\pi_1)$. Similarly, let $\pi_2 = \pi_{x_2}$, $p_2 = \cpat{\pi_2}{x_2}$, $c_2 = \ccol{\pi_2}{x_2}$, $\Pi_2 = \Pi_{x_2}(\pi_2)$, and $P_2 = P_{x_2}(\pi_2)$.
        It holds by \cref{def:sol-seq} that 
        $\pi = \pi_1 \dot\cup \pi_2$, and hence, by the definition of $\Pi$, that 
        \begin{equation}\label{eq:acseq-repr-sol:union-node-Pi}            
            \Pi = \{\tau_1\dot\cup\tau_2 \colon \tau_1 \in \Pi_1 \land \tau_2 \in \Pi_2\}.    
        \end{equation}
        It follows that $\cpat{\pi}{x} = \cpat{\pi_{x_1}}{x_1}\punion\cpat{\pi_{x_2}}{x_2}$ and from \cref{def:ac-seq}, that $P = P_1 \punion P_2$. By induction hypothesis, it holds that $P_1$ represents $p_1$, and $P_2$ represents $p_2$. Hence, by \cref{lem:ops-pres-rep} it holds that $P$ represents $p$. Also, by the induction hypothesis, it holds for each $i \in \{1,2\}$, and for each $\tau_i \in P_i$ that $\tau_i$ generates the coloring $c_i$ at $x_i$ and that $\tau_i$ is valid if and only if $\pi_i$ is.
        
        From \cref{eq:acseq-repr-sol:union-node-Pi}, it holds for $\tau \in \Pi$ that $\tau = \tau_1 \dot\cup \tau_2$ for some $\tau_1\in \Pi_1$ and $\tau_2\in\Pi_2$. Hence, it follows that $\tau$ generates the coloring $c_1 \cjoin c_2 = c$, and $\tau$ is valid, if and only if both $\tau_1$ and $\tau_2$ are, which is the case if and only if both $\pi_1$ and $\pi_2$ are, which holds if and only if $\pi$ is valid.
    \end{itemize}
\end{proof}

\begin{lemma}\label{lem:conn-if-consistent-zero}
    Let $x\in \nodes$, and $S\subseteq V_x$ be a non-empty set of vertices. Then it holds that $p=\solpat{S}{x}\sim[0]$ if and only if $G_x[S]$ is connected, and $v_0 \in S$.
\end{lemma}

\begin{proof}
    It holds that $p\sim[0]$ if and only if $p$ consists of the zero-set only. However, since $S$ is not empty, it holds that $G_x[S]$ contains at least one connected component. Hence, $p$ consists of the zero-set only, if and only if $v_0 \in S$ and $G_x[S]$ is connected.
\end{proof}

\begin{corollary}\label{cor:sol-if-ac-seq}
    Let $\budget \in [n]$. There exists a weight $\weight$, a coloring $c\in \Coloring$ and a valid action-sequence $\tau$ of cost $\budget$ and weight $\weight$ at $r$ generating the pair $([0], c)$, if and only if there exists a connected odd cycle transversal of size $\budget$ in $G$ containing $v_0$.
\end{corollary}

\begin{proof}
    Let $S\subseteq V$ with $|S|=\budget$. It holds that $S$ is a connected odd cycle transversal of $G$ containing $v_0$ if and only if $G[S]$ is connected, $v_0\in S$, and there exists a proper $2$-coloring $g:V\setminus S\rightarrow \{\black,\white\}$ of $G\setminus S$.
    It holds by \cref{lem:conn-if-consistent-zero} that the former two conditions are the case if and only if $\solpat{S}{r} \sim [0]$. By \cref{lem:action-seq-rep-sol} it follows that the latter is the case if and only if there exists a valid action-sequence at $r$ of cost $\budget$ and some weight $\weight$ generating a pair $(p', c)$ where $c=\solcol{g}{r}$, and $p'\sim [0]$. Since $p'$ is a complete pattern, it holds due to \cref{lema:zero-only-complete-consistent} that $p' = [0]$.    
\end{proof}

\subsection{The Algorithm counts action-sequences}

In this section, we show that using the tables $T$ defined in \cref{sec:algo}, one can count (modulo 2) the number of action-sequences of cost $\budget$ and weight $\weight$, that generate patterns consistent with a complete pattern $q \in \Cp$.

\begin{definition}
    We say that a set of patterns $S$ \emph{parity-represents} another set of patterns $T$ over a family of patterns $\Pat^* \subseteq \Pat$, if for each pattern $q \in \Pat^*$ it holds that 
    \[
        |\{p\in S\colon p\sim q\}| \bquiv |\{p\in T\colon p\sim q\}|.
    \]
    We say that $S$ \emph{parity-represents} a pattern $p$ over $\Pat^*$, if $S$ parity-represents $\{p\}$ over $\Pat^*$.
\end{definition}

\begin{lemma}[{\cite[Lemma 8.10]{DBLP:journals/corr/abs-2307-14264/BojikianK23}}] \label{lem:parrep-parity-rep}
    Let $p\in \Cp$. The family $\parrep(p)$ parity-represents $p$ over $\Cp$.
\end{lemma}

\begin{definition}
    Given a $k$-ary operation over patterns $\op\colon \Pat^k \rightarrow \Pat$, we say that $\op$ preserves $\Pat^*$-parity-representation over $\hat{\Pat}$ for two families of patterns $\Pat^*, \hat{\Pat}\subseteq \Pat$ if and only if
    for all sets $S_1,\dots, S_k, T_1,\dots T_k \subseteq \hat{\Pat}$ such that $T_i$ parity-represents $S_i$ over $\Pat^*$ for all $i\in[k]$ it holds that $\exec(S_1,\dots, S_k)$ parity represents $\exec(T_1,\dots, T_k)$ over $\Pat^*$.
\end{definition}

\begin{lemma}[{\cite[Lemma 8.12]{DBLP:journals/corr/abs-2307-14264/BojikianK23}}]\label{lem:ops-pres-parity-rep}
    The operations join, relabel and union, as defined in \cref{def:patops}, preserve $\Cp$-parity-representation over $\Cp$.
\end{lemma}

\begin{lemma}[{\cite[Lemma 8.14]{DBLP:journals/corr/abs-2307-14264/BojikianK23}}] \label{lem:actions-pres-parity-rep}
    Let $S,T \subseteq \Cp$ be two families of complete patterns, such that $T$ parity-represents $S$ over $\Cp$, and $i,j\in[k]$.
    Then it holds for all $\ell\in[4]$ that $\exac(\patadd_{i,j}T,\ell)$ parity-represents $\exac(\patadd_{i,j}S,\ell)$ over $\Cp$.
\end{lemma}

\begin{lemma} \label{lem:vectors-pres-parity-rep}
    Let $\op$ be some $k$-ary operation over patterns that preserves $\Pat^*$-parity-representation over $\hat{\Pat} \subseteq \Pat$ for two families $\Pat^*, \hat{\Pat}$. Let $P_1,P_2\subseteq \hat{\Pat}$ be two subfamilies, both closed under $\op$. Let $S_1, \dots S_k \in \{0,1\}^{P_1}$, and $T_1, \dots T_k\in\{0,1\}^{P_2}$. If it holds for each $i\in[k]$, and for each $q\in \Pat^*$ that 
    \[
        \sumstack{p\in P_1\\p\sim q}S_i\ind{p}\bquiv \sumstack{p\in P_2\\p\sim q}T_i\ind{p},
    \]
    then for $S^*\in \{0,1\}^{P_1}$, with 
    \[S^*\ind{p}= \sumstack{(p_1,\dots, p_k)\in P_1^k\\
    \op(p_1,\dots,p_k)=p}
    \Big(S_1\ind{p_1}\cdot S_2\ind{p_2}\dots S_k\ind{p_k}\Big),\]
    and for $T^*\in \{0,1\}^{P_2}$, with
    \[T^*\ind{p}= \sumstack{(p_1,\dots, p_k)\in  P_2^k\\
    \op(p_1,\dots,p_k)=p}
    \Big(T_1\ind{p_1}\cdot T_2\ind{p_2}\dots T_k\ind{p_k}\Big),\]
    it holds for all $q\in P^*$ that
    \[
    \sumstack{p\in P_1\\p\sim q}S^*\ind{p}\bquiv \sumstack{p\in P_2\\p\sim q}T^*\ind{p},
    \]
\end{lemma}

\begin{proof}
    We define the sets $X_1,\dots,X_k \subseteq P_1$ as 
    \[X_i = \{p \in P_1\colon S_i\ind{p} = 1\}.\]
    Similarly, we define $Y_1,\dots,Y_k \subseteq P_2$ as 
    \[Y_i = \{p \in P_2\colon T_i\ind{p} = 1\}.\]
    For $q\in\Pat^*$ and $i\in[k]$ it holds that 
    \[|\{p\in X_i\colon p\sim q\}| \bquiv \sumstack{p\in P_1\\p\sim q}S_i\ind{p} \bquiv \sumstack{p\in P_2\\p\sim q}T_i\ind{p} \bquiv |\{p\in Y_i\colon p\sim q\}|\]
    Hence, it holds that $X_i$ parity-represents $Y_i$ over $\Pat^*$ for each $i\in [k]$. Let $X^* = \exec(X_1,\dots, X_k)$, and $Y^* = \exec(Y_1,\dots,Y_k)$. Since $\op$ preserves $\Pat^*$-parity-representation over $\hat{\Pat}$, it holds that $X^*$ parity-represents $Y^*$ over $\Pat^*$.
    We claim that 
    \[X^* = \{p\in P_1\colon S^*\ind{p} = 1\},\]
    and
    \[Y^* = \{p\in P_2\colon T^*\ind{p} = 1\}\]
    hold. Assuming this is the case, it follows for $q\in \Pat^*$ that 
    \begin{align*}
    \sumstack{p\in P_1\\p\sim q}S^*\ind{p} &\bquiv |\{p\in X^*\colon p\sim q\}|\\
    &\bquiv |\{p\in Y^*\colon p \sim q\}\bquiv \sumstack{p\in P_2\\p\sim q}T^*\ind{p}.
    \end{align*}
    We only prove the former claim, the latter follows analogously, by exchanging $X^*, X_i, S^*, S_i, P_1$ with $Y^*, Y_i, T^*, T_i, P_2$ respectively. It holds for $p\in P_1$ that
    \begin{align*}
        p \in X^* &\iff p \in \bigdelta\limits_{(p_1,\dots p_k)\in X_1\times\dots\times X_k}
            \big\{\op(p_1,\dots, p_k)\big\}\\
        &\iff \sum\limits_{(p_1,\dots p_k)\in X_1\times\dots\times X_k}
        [\op(p_1,\dots, p_k)=p] \bquiv 1\\
        &\iff \sum\limits_{(p_1,\dots p_k)P_1^k}
        [p_1\in X_1]\cdot \ldots \cdot [p_k\in X_k]\cdot [\op(p_1,\dots, p_k)=p] \bquiv 1\\
        &\iff \sumstack{(p_1,\dots p_k)\in P_1^k\\ \op(p_1,\dots, p_k)=p}
        S_1\ind{p_1} \cdot \ldots \cdot S_k\ind{p_k} \bquiv 1\\
        &\iff S^*\ind{p} \bquiv 1.\\
    \end{align*}
\end{proof}

\begin{lemma}\label{lem:alg-counts-sequences}
    For each node $x\in V(\syntaxtree)$, for all values $\budget, \weight$, and all $c \in \Coloring$ and $q\in\Cp$, it holds that 
    \[
        \sumstack{p\in\CSP\\p\sim q}T\nodeind\ind{(p,c)}\bquiv 1,
    \]
    if and only if there exists an odd number of valid action-sequences $\tau$ at $x$ of cost $\budget$ and weight $\weight$ generating the pair $(p,c)$, where $p\sim q$.
\end{lemma}

\begin{proof}
    For $p\in\Cp$ and $c\in\Coloring$, let $D\nodeind\ind{(p,c)}$ be the number of valid action-sequences at $x$ of cost $\budget$ and weight $\weight$ generating the pair $(p,c)$.
    We prove by induction over $\syntaxtree$ that for all $x\in\nodes$, for all values $\budget,\weight$, for all $c\in \Coloring$ and all complete patterns $q\in \Cp$ it holds that
    \[
        \sum\limits_{\substack{p\in\CSP\\p\sim q}}T\nodeind\ind{(p,c)}\bquiv
        \sum\limits_{\substack{p\in\Cp\\p\sim q}}D\nodeind\ind{(p,c)}.
    \]
    \begin{itemize}
        \item Introduce node $\mu_x = i(v)$: Clearly there exist four valid action-sequences at $x$, each generating a complete pattern that is a $CS$-pattern as well. Hence, it holds for $p\in \CSP$ that 
        \[D\nodeind\ind{(p,c)} = T\nodeind\ind{(p,c)},\]
        and $D\nodeind\ind{(p,c)} = 0$ for $p\in \Cp\setminus \CSP$.
            
        \item Relabel node $\mu_x = \rho_{i\rightarrow j}(\mu_{x'})$:
        Let $(p, c) \in \Cp \times \Coloring$. For each pair $(p', c')\in \Cp\times \Coloring$ with $(p')_{i\rightarrow j} = p$ and $(c')_{i\rightarrow j} = c$ it holds that each valid action-sequence at $x'$ generating $(p', c')$, is a valid action-sequence at $x$ that generates $(p, c)$. It also holds for each valid action-sequence $\tau$ at $x$ generating $(p, c)$, that $\tau_{x'}$ generates a pair $(p', c')$ at $x'$ where $(p')_{i\rightarrow j} = p$ and $(c')_{i\rightarrow j} = c$. Hence, it holds that
        \[
         D\nodeind\ind{(p,c)} = \sum\limits_{\substack{(p',c')\in\Cp\times\Coloring,\\p'_{i\rightarrow j} = p \land c'_{i\rightarrow j} = c}} D\ind{x',\budget,\weight}\ind{(p',c')}.   
        \] 
        By induction hypothesis, it holds for all values $\budget,\weight$, and all $c\in \Coloring$ and $q\in\Cp$ that 
        \[
            \sumstack{p\in\Cp\\p\sim q}D\subnodeind\ind{(p, c)}\bquiv
            \sumstack{p\in\CSP\\p\sim q}T\subnodeind\ind{(p, c)}.
        \]
        It holds by \cref{lem:ops-pres-parity-rep} that $i\rightarrow j$ preserves $\Cp$-parity-representation over $\Cp$ (and hence, over $\CSP$).
        Hence, it holds by \cref{lem:vectors-pres-parity-rep} and the definition of $T\nodeind$ for a relabel node $x$, for $c\in \Coloring$ and $q\in \Cp$ that
        \[
            \sumstack{p\in\Cp\\p\sim q}D\nodeind\ind{(p, c)}\bquiv
            \sumstack{p\in\CSP\\p\sim q}T\nodeind\ind{(p, c)}.
        \]
     
        \item Union node $\mu_x = \mu_{x_1} \union \mu_{x_2}$:
        Let $(p, c) \in \Cp \times \Coloring$. For all pairs $(p_1, c_1), (p_2,c_2)\in \Cp\times \Coloring$ with $p_1 \punion p_2 = p$ and $c_1 \cjoin c_2 = c$ it holds that for each two valid action-sequences $\tau_1$ generating $(p_1, c_1)$ at $x_1$, and $\tau_2$ generating $(p_2, c_2)$ at $x_2$, that $\tau=\tau_1 \dot\cup \tau_2$ is a valid action-sequence at $x$ that generates $(p, c)$. Also, for each valid action-sequence $\tau$ at $x$ generating $(p, c)$, there exist two tuples $(p_1, c_1), (p_2,c_2)$ with  $p_1 \punion p_2 = p$ and $c_1 \cjoin c_2 = c$, such that $\tau_{x_1}$ generates $(p_1, c_1)$ at $x_1$ and $\tau_{x_2}$ generates $(p_2, c_2)$ at $x_2$. Hence, it holds that
        \[
        D\nodeind\ind{(p,c)} =
        \sumstack{\budget_1+\budget_2=\budget\\\weight_1+\weight_2=\weight}
        \Big(\sumstack{p_1\punion p_2 = p\\c_1\cjoin c_2=c}
        D\leftind\ind{(p_1,c_1)} \cdot D\rightind\ind{(p_2,c_2)}\Big).
        \]
        By induction hypothesis, it holds for $x'\in\{x_1, x_2\}$ and all $\budget,\weight$, $c\in \Coloring$ and $q\in\Cp$ that 
        \[
            \sumstack{p\in\Cp\\p\sim q}D\subnodeind\ind{(p,c)}\bquiv
            \sumstack{p\in\CSP\\p\sim q}T\subnodeind\ind{(p,c)}.
        \]

        It holds by \cref{lem:ops-pres-parity-rep} that $\punion$ preserves $\Cp$-parity-representation over $\Cp$ (and hence, over $\CSP$).
        Hence, it holds by \cref{lem:vectors-pres-parity-rep} and the definition of $T\nodeind$ for a union node $x$, for $c\in \Coloring$ and $q\in \Cp$ that
        \[
            \sumstack{p\in\Cp\\p\sim q}D\nodeind\ind{(p,c)}\bquiv
            \sumstack{p\in\CSP\\p\sim q}T\nodeind\ind{(p,c)}.
        \]
        
        \item Join node $\mu_x = \eta_{i,j} (\mu_{x'})$:
        For $c\in\Coloring$ with $c(i)\not\sim c(j)$, it holds clearly for each complete pattern $p$, that $D\nodeind\ind{(p,c)} = 0$.
        Since it also holds that $T\nodeind\ind{(p,c)} = 0$ for each $p\in\CSP$ by definition, the claim follows in this case.
        Now assume that $c(i)\sim c(j)$.
        For $\ell \in [4]$, let $\tau$ be a valid action-sequence at $x$ generating the pair $(p,c)$ for some $p\in \Cp$, such that $\tau(x)= \ell$. It must hold that $\tau_{x'}$ is a valid action-sequence at $x'$ generating the pair $(p', c)$ with $\action(\patadd_{i,j}p', \ell) = p$. Moreover, since we assume that $c(i)\sim c(j)$, each valid action-sequence $\tau'$ at $x'$ generating the pair $(p', c)$ for some $p'\in\Cp$ can be extended by $x\mapsto \ell$ to a valid action-sequence at $x$ generating the pair $(p, c)$ with $p = \action(\patadd_{i,j}p', \ell)$. Hence, it holds that 
        \[
        D\nodeind\ind{(p,c)} =
        \sum\limits_{\ell\in[4]}
        \sumstack{p'\in\Cp\\ \action(\patadd_{i,j}p', \ell)=p}
        D\ind{x',\budget,\weight - \weightf(x, \ell)}\ind{(p',c)}.
        \]
        By induction hypothesis, it holds for all values $\budget,\weight$ and for all $q\in \Cp$ that
        \[
            \sumstack{p\in\Cp\\p\sim q}D\subnodeind\ind{(p, c)}\bquiv
            \sumstack{p\in\CSP\\p\sim q}T\subnodeind\ind{(p, c)}.
        \]

        We define the tables $T'\subnodeind\in\{0,1\}^{\Cp\times \Coloring}$ as
        \[
            T'\subnodeind\ind{(p,c)} = 
            \begin{cases}
                T\subnodeind\ind{(p,c)} &\colon p\in\CSP,\\
                0 &\colon \text{otherwise.}
            \end{cases}  
        \]
        Then clearly, it holds for all values $\budget, \weight$, and all $q\in\Cp$ that
        \[
            \sumstack{p\in\Cp\\p\sim q}D\subnodeind\ind{(p, c)}\bquiv
            \sumstack{p\in\Cp\\p\sim q}T'\subnodeind\ind{(p, c)}.
        \]
        Finally, let $T''\nodeind\in\{0,1\}^{\Cp \times \Coloring}$ defined by 
        \begin{align*}
            T''\nodeind\ind{(p,c)} &=
            \sum\limits_{\ell\in[4]}
            \sumstack{p'\in\Cp\\ \action(\patadd_{i,j}p', \ell)=p}
            T'\ind{x',\budget,\weight - \weightf(x, \ell)}\ind{(p',c)}\\
            &=
            \sum\limits_{\ell\in[4]}
            \sumstack{p'\in\CSP\\ \action(\patadd_{i,j}p', \ell)=p}
            T\ind{x',\budget,\weight - \weightf(x, \ell)}\ind{(p',c)}.
        \end{align*}
        It holds by \cref{lem:actions-pres-parity-rep} that actions preserve $\Cp$-parity-representation over $\Cp$, and hence, by \cref{lem:vectors-pres-parity-rep} that 
        \[
            \sumstack{p\in\Cp\\p\sim q}D\nodeind\ind{(p, c)}\bquiv
            \sumstack{p\in\Cp\\p\sim q}T''\nodeind\ind{(p, c)}.
        \]
        Finally, we show that 
        \[
        \sumstack{p\in\Cp\\p\sim q}T''\nodeind\ind{(p, c)}\bquiv
        \sumstack{p\in\CSP\\p\sim q}T\nodeind\ind{(p, c)},
        \]
        which is equivalent to
        \begin{align*}
        \sumstack{p\in\Cp\\p\sim q}
        \sum\limits_{\ell\in [4]}
        \sumstack{p'\in\CSP\\ \action(\patadd_{i,j}p', \ell)=p}
        &T\ind{x',\budget,\weight - \weightf(x, \ell)}\ind{(p',c)}\bquiv\\
        &\sumstack{p\in\CSP\\p\sim q}
        \sum\limits_{\ell\in [4]}
        \sumstack{p'\in\CSP\\ p\in \parrep\big(\action(\patadd_{i,j}p', \ell)\big)}
        T\ind{x',\budget,\weight - \weightf(x, \ell)}\ind{(p',c)}.
        \end{align*}
        This holds, if for each $\ell \in [4]$ it holds that
        \[
            \sumstack{p'\in\CSP\\ \action(\patadd_{i,j}p', \ell)\sim q}
            T\ind{x',\budget,\weight - \weightf(x, \ell)}\ind{(p',c)}\bquiv\\
            \sumstack{p'\in\CSP,\\ p\in \parrep\big(\action(\patadd_{i,j}p', \ell)\big),\\p\sim q}
            T\ind{x',\budget,\weight - \weightf(x, \ell)}\ind{(p',c)},
        \]
        which is the case, if for each $p'\in \CSP$ and $p=\action(\patadd_{i,j}p', \ell)$ it holds that $p\sim q$ if and only if $|\{ r\in\parrep(p)\colon p\sim q\}|$ is odd, which is clearly the case due to \cref{lem:parrep-parity-rep}.
    \end{itemize}
\end{proof}

\begin{corollary}\label{cor:odd-sequences-implies-table}
    For two values $\budget, \weight$, and a coloring $c\in\Coloring$, it holds that $T\rootind\ind{([0], c)} = 1$ if and only if there exists an odd number of valid action-sequences of cost $\budget$ and weight $\weight$ at $r$ generating $([0],c)$.
\end{corollary}

\begin{proof}
    By \cref{lem:alg-counts-sequences}, it holds that
    \[\sumstack{p\in\CSP\\p\sim [0]}T\rootind\ind{(p,c)} \equiv_2 1,\]
    if and only if there exists an odd number of valid action-sequences $\tau$ at $r$ of cost $\budget$ and weight $\weight$ generating a pair $(p,c)$ with $p\sim [0]$.
    It holds by \cref{lema:zero-only-complete-consistent} that $[0]$ is the only complete pattern consistent with $[0]$, and hence, it holds that 
    \[T\rootind\ind{([0], c)} = \sumstack{p\in\CSP\\p\sim [0]}T\rootind\ind{(p,c)} \equiv_2 1,\]
    if and only if there exists an odd number of valid action-sequences $\tau$ at $r$ of cost $\budget$ and weight $\weight$ generating the pair $([0],c)$.
\end{proof}

\subsection{Isolation lemma}\label{sec:iso}

\begin{definition}\label{def:isolation}
    A weight function $\weightf:U\rightarrow \mathbb{Z}$ \emph{isolates} a set family $\mathcal{F}\subseteq 2^U$, if there exists a unique $S'\subseteq \mathcal{F}$ with $\weightf(S') = \min_{S\in\mathcal{F}}\weightf(S)$.
\end{definition}

\begin{lemma}[\cite{DBLP:journals/combinatorica/MulmuleyVV87}]\label{lem:iso}
    Let $\mathcal{F}\subseteq \pow{U}$ be a set family over a universe $U$ with $|\mathcal{F}|>0$, and let $N>|U|$ be an integer. For each $u\in U$, choose a weight $\omega(u)\in \{1,2,\dots N\}$ uniformly and independently at random. Then it holds that $\pr[\omega \text{ isolates } \mathcal{F}]\geq 1-|U|/N$.
\end{lemma}

\begin{lemma}\label{lem:iso-imply-single-sol}
    Assume that $G$ contains a connected odd cycle transversal of size $\target$.
    Fix $\W = 8 \cdot |\wnodes|$, and let $\weightf$ is chosen uniformly at random in $[\W]^{\big(\wnodes \times [4]\big)}$. Let $\tau$ be a minimum weight valid action-sequence of cost $\target$ generating a pair $([0], c)$ for any coloring $c$. Then $\tau$ is unique with probability at least one half.
\end{lemma}

\begin{proof}
    Let $P$ be the family of all valid action-sequences of cost $\target$ generating a pair $([0], c)$ for any coloring $c$. For each $\tau \in P$, the weight of $\tau$ is given by $\sum\limits_{x\in\wnodes}\weightf(x,\tau(x))$. Let $U = \wnodes\times[4]$, and for $\tau\in P$, let $F_{\tau} = \{(x,\tau(x))\colon x \in \wnodes\}$. Let $\mathcal{F} = \{F_{\tau}\colon \tau \in P\}$. Then $\tau\mapsto F_{\tau}$ is a bijection from $P$ to $\mathcal{F}$. Moreover, it holds that $\weightf(\tau) = \weightf(F_{\tau})$ for all $\tau \in P$. Finally, by \cref{lem:iso}, it holds that $\weightf$ isolates $\mathcal{F}$ with probability
    \begin{align*}
        \pr[\weightf \text{ isolates } \mathcal{F}] &\geq 1-4|\wnodes| / 8|\wnodes|\\
        & = 1 - 1/2 = 1/2.
    \end{align*}
\end{proof}

\subsection{Correctness of the algorithm}

\begin{proof}[Proof of \cref{theo:upper-bound}]
	Let $\W = 8 \cdot |\wnodes|$. First, choose a weight function $\weightf:\wnodes\times [4] \rightarrow \W$, by choosing $\weightf(x,\ell) \in [\W]$ independently and uniformly at random for each $(x, \ell)\in\wnodes\times[4]$. For each vertex $v\in V$ the algorithm repeats the following: the algorithm chooses $v_0 = v$, and then computes all tables $T\nodeind$ for all $x\in\nodes$, and all values $\budget \in [\target]$ and $\weight \in [4\cdot|\wnodes|\cdot\W]$ in a bottom-up manner over $\syntaxtree$. The algorithm returns yes, if for any choice of $v_0$ there exists a coloring $c\in\Coloring$ and a weight $\weight$ such that $T\solind\ind{([0],c)} = 1$. Otherwise, it returns no.

    If the graph does not contain a connected odd cycle transversal of size $\target$, then it holds for any choice of $v_0\in V$, by \cref{cor:sol-if-ac-seq}, that there exists no action-sequence of cost $\target$ and weight $\weight$ at $r$, that generates the pair $([0], c)$ for any weight $\weight$ and coloring $c$. By \cref{cor:odd-sequences-implies-table} it holds that $T\solind\ind{([0], c)} = 0$ for all values $\weight$ and any coloring $c$.
    
    On the other hand, if a solution of size $\target$ exists. Let $v_0$ be any vertex in any solution of size $\target$. Then by \cref{lem:iso-imply-single-sol} there exists a weight $\weight$ and a coloring $c$, such that with probability at least $1/2$ there exists a unique action-sequence at $r$ of cost $\target$ and weight $\weight$ generating the pair $([0], c)$, which implies by \cref{cor:odd-sequences-implies-table} that $T\solind\ind{([0], c)} = 1$. Hence, with probability at least $1/2$ the algorithm outputs yes.
\end{proof}

\section{Lower bound}\label{sec:lb}

Let $\var = \{v_1,\dots v_n\}$ be a set of variables.
In the $d$-SAT problem, an instance $I$ is a set of clauses $I = \{C_1,\dots C_m\}$, where each clause $C_j = \{a^j_1,\dots a^j_d\}$ is a set of $d$ literals. A literal $a$ is either of the form $a = v_i$ (a positive literal) or $a = \overline{v_i}$ (a negative literal). Let $\var(a) = v_i$ be the variable underlying the literal $a$, and for a clause $C$ we define $\var(C) = \{\var(a) \colon a \in C\}$. A partial assignment $\alpha:V'\rightarrow\{\true, \false\}$ for $V'\subseteq V$ is a Boolean mapping. We say that $\alpha$ satisfies a literal $a$, with $\var(a) = v$, if $v \in V'$ and it holds that $\alpha(v) = \true$ if and only if $a$ is a positive literal. We say that $\alpha$ satisfies a clause $C$ if it satisfies at least one literal of $C$, and $\alpha$ satisfies $I$ if it satisfies each clause $C_j$ of $I$. 
The goal of the $d$-SAT problem is to find a satisfying assignment $\alpha$ of $I$ or declare that no such assignment exists.

The Strong Exponential Time Hypothesis (SETH) \cite{DBLP:journals/jcss/ImpagliazzoP01,DBLP:journals/jcss/ImpagliazzoPZ01} conjectures that the trivial algorithm that tests all different assignments to find a satisfying one using brute-force is essentially optimal for general values of $d$. In the following, we provide a formal statement of the hypothesis:

\begin{conjecture}[\cite{DBLP:journals/jcss/ImpagliazzoP01,DBLP:journals/jcss/ImpagliazzoPZ01}]
    For any positive value $\delta$ there exists an integer $d$ such that the $d$-SAT problem cannot be solved in time $\ostar((2-\delta)^n)$, where $n$ denotes the number of the variables.
\end{conjecture}

Our lower bound follows similar tight lower bounds for structural parameters \cite{DBLP:journals/talg/LokshtanovMS18, DBLP:journals/corr/abs-1103-0534/CyganNPPRW11, DBLP:conf/iwpec/HegerfeldK22} based on SETH. We provide a parameterized reduction from the $d$-SAT problem to the \Coctp problem, for any value of $d$, bounding the value of the parameter in the resulting graph. We show that an algorithm with running time $\ostar((12-\epsilon)^{\lcw})$ for any positive value $\epsilon$ implies an algorithm with running time $\ostar((2-\delta)^{n})$ for a positive value $\delta >0$ for the $d$-SAT problem for any value of $d$. This would contradict SETH, and hence, assuming SETH, the lower bound must hold. In the following, we describe the instance $(G, \budget)$ that results from the reduction given an instance $I$ of the $d$-SAT problem.

\subsection{Construction of the graph}
By adding the triangle $\tri{u}{v}$ to a graph $G$ for two vertices $u,v\in V(G)$, we mean that we add the edge $\{u,v\}$, and we add a new vertex $w_{u,v}$ and make it adjacent to $u$ and $v$ only. When we say we add a triangle at a vertex $v$, we mean that we add two new vertices $u_v, w_v$, and make them adjacent to each other and to $v$ only. Vertices introduced when adding triangles are called triangle vertices. For a set of vertices $S\subseteq V$, we denote by $\inctri{S}$ the set consisting of $S$ and all vertices added to create triangles at single vertices of $S$ or between pairs of vertices in $S$.

We build $G$ by combining copies of fixed components of constant sizes, called gadgets. We distinguish three kinds of gadgets in $G$, namely \emph{path gadgets}, \emph{decoding gadgets}, and \emph{clause gadgets}. We distinguish a single vertex $r$ in $G$ called the \emph{root} vertex, with a triangle added to it, and another single vertex $\broot$. When we say that we color a vertex \emph{white} in $G$, it means that we add an edge between this vertex and $\broot$. When we say we color a vertex $u$ \emph{black}, we mean that we introduce a new vertex $u'$, we color it white, and we add an edge between $u$ and $u'$. We also distinguish two vertices $g_1$ and $g_2$ in $G$, called \emph{guards}. We add a triangle at each of them.

Let $t_0$ be a constant that we fix later, and $s = \left\lceil n/t_0 \right\rceil$. We partition the variables into $s$ groups $V_1, \dots V_s$ each of size $t_0$ (except for the last one). Hence, the $i$th group contains the variables $v_{(i-1)\cdot t_0 + 1},\dots, v_{i\cdot t_0}$, except for the last group, where we replace $i\cdot t_0$ with $n$.

Let $t = \left\lceil t_0 \cdot \log_{12}(2) \right\rceil$ and $n' = s\cdot t$. The graph $G$ consists of $n'$ path-like structures, called \emph{path sequences}. Each path sequence is a sequence of $c = (11\cdot n' + 1)\cdot m$ path gadgets. Consecutive path gadgets are connected by cuts of size $9$, given as bicliques between three vertices of one gadget, and three of the other, giving a path sequence its narrow path-like structure. The first path gadget of each path sequence is adjacent to $g_1$, while the last path gadget is adjacent to $g_2$.

Let $P_1, \dots, P_{n'}$ be the set of all path sequences in $G$. We group them into $s$ groups of size $t$ each, calling each of these groups a \emph{path bundle}. Hence, the $i$th path bundle $\pbundle{i}$ consists of the path sequences $P_{(i-1)\cdot t + 1} \dots P_{i\cdot t}$, and it corresponds to the variable group $V_i$. For each $i\in [s], j \in [c]$, we call the set of all $j$th path gadgets on all path sequences in the path bundle $\mathscr{B}_i$ a \emph{(simple) bundle} $\bundle{i}{j}$.

For a path sequence $P_i$, let $\pgad{i}{1},\dots \pgad{i}{c}$ be the path gadgets on the path sequence $P_i$ from left to right. We divide each path sequence $P_i$ into $(11\cdot n' + 1)$ groups of size $m$ each, called \emph{segments} $\seg{i}{1},\dots,\seg{i}{11\cdot n' + 1}$, where $\seg{i}{j} = \pgad{i}{(j-1)\cdot m + 1}, \dots, \pgad{i}{j\cdot m}$. For $j\in[11\cdot n' + 1]$, we call the set consisting of the $j$th segment of each path sequences a \emph{section} $\sect{j}$, i.e.\ $\sect{j} = \{\seg{i}{j}\colon i\in[n']\}$. Finally, for $j\in [c]$, we call the set consisting of the $j$-th path gadget of each path sequence a \emph{column} $\col{j} = \{\pgad{i}{j}\colon i\in[n']\}$. See \Cref{fig:lb-construction} for illustration.

\begin{figure}
    \centering
    \begin{subfigure}[b]{0.3\textwidth}
        \centering
        \includegraphics[width=\textwidth]{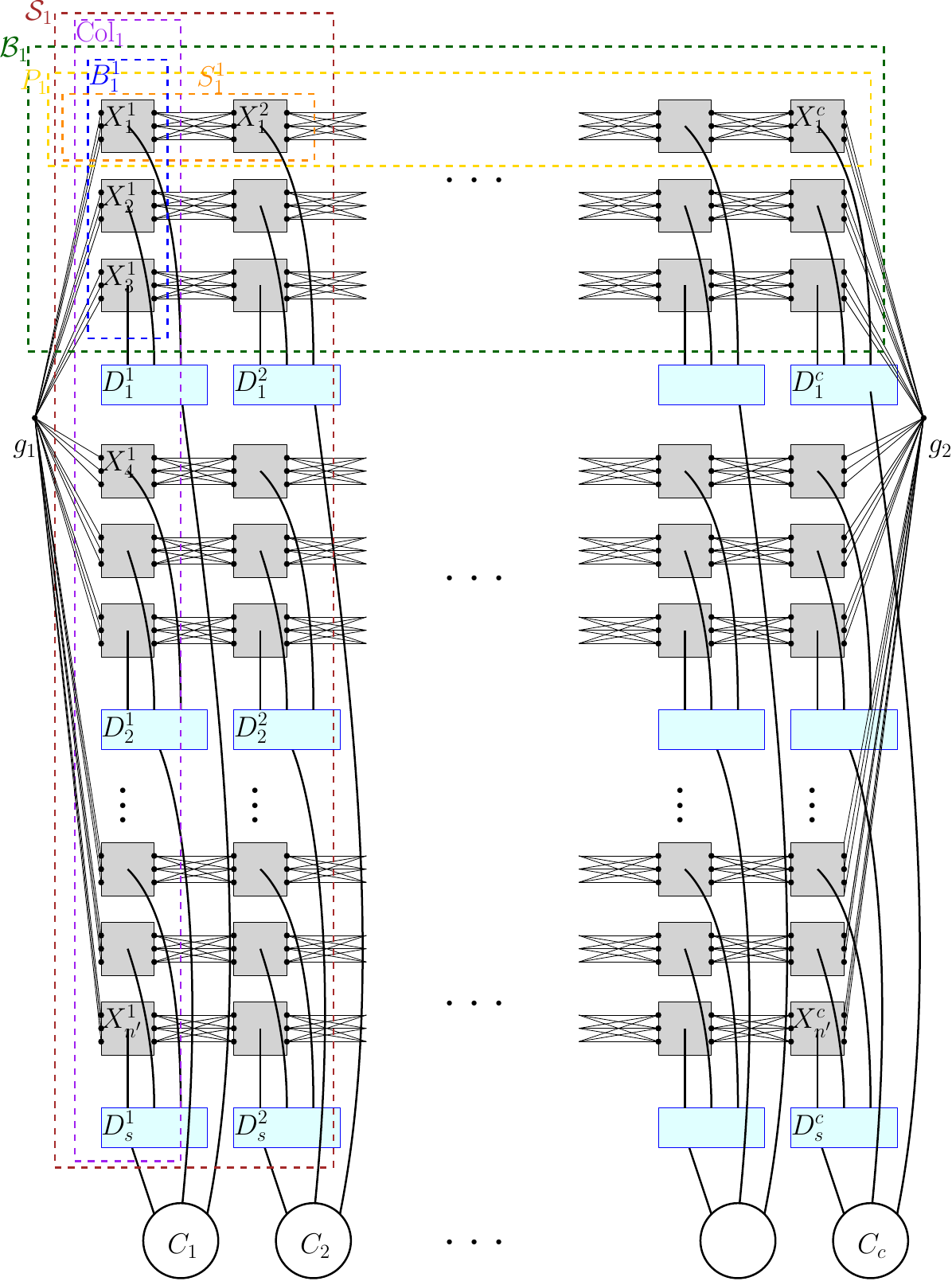}
        
    \end{subfigure}
    \hfill
    \begin{subfigure}[b]{0.3\textwidth}
        \centering
        \includegraphics[width=\textwidth]{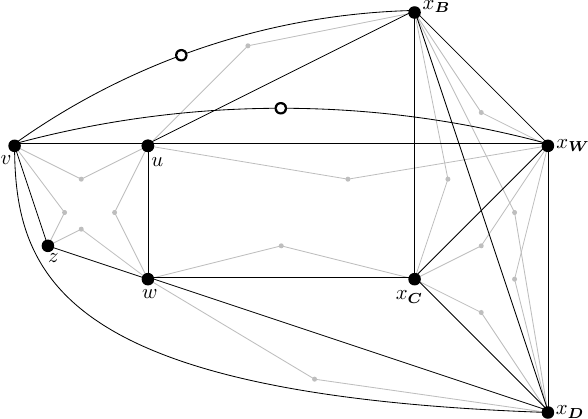}
    \end{subfigure}
    \caption{On the left, the graph $G$ corresponding to an instance with  two clauses. The gray squares represent path gadgets. Decoding gadgets are depicted in cyan, and clause gadgets as circles. We outline in yellow, blue, green, purple, orange and brown, the first path sequence, bundle, path bundle, column, segment and section respectively. On the right we depict a simple path gadget. Subdivision vertices are drawn as circles, while triangle vertices are drawn as small gray disks.
    \label{fig:lb-construction}}
\end{figure}

We attach a decoding gadget $\dgad{i}{j}$ to each bundle $\bundle{i}{j}$ for $i\in[s]$ and $j\in[c]$, and we attach a clause gadget $\cgad{j}$ to each column $\col{j}$ for $j\in[c]$. A decoding gadget is only adjacent (possibly through triangles) to vertices of the path gadgets in its corresponding bundle, to the clause gadget in its column, and to $r$. A clause gadget is only adjacent to vertices in the decoding gadgets attached to bundles in its column.

We define the set $\States = \{\conn, \disc, \black, \white\}$, and call it the set of \emph{simple  states}. We define the set of \emph{states} $\pstates$ as the set of tuples $\pstates = \{s_1, \dots s_{12}\}\subseteq \States^3$ given by the following list:

\begin{center}
\begin{tabularx}{.85\textwidth}{X X l}    
    \bp $s_{1} =  (\conn,\conn,\conn)$&
    \bp $s_{2} =  (\conn,\conn,\disc)$&
    \bp $s_{3} =  (\conn,\disc,\black)$\\
    \bp $s_{4} =  (\conn,\disc,\white)$&
    \bp $s_{5} =  (\conn,\black,\white)$&
    \bp $s_{6} =  (\disc,\disc,\disc)$\\
    \bp $s_{7} =  (\disc,\disc,\black)$&
    \bp $s_{8} =  (\disc,\disc,\white)$&
    \bp $s_{9} =  (\disc,\black,\white)$\\  
    \bp $s_{10} = (\black,\black,\black)$&
    \bp $s_{11} = (\white,\white,\white)$&
    \bp $s_{12} = (\black,\black,\white)$\\
\end{tabularx}
    \end{center}

Intuitively, a solution $\sol$ defines a state $s\in\pstates$ in each path gadget, given by the intersection of $\sol$ with this path gadget. A combination of states over a whole bundle $\bund = \{\pgadget_1, \dots \pgadget_t\}$, given from the intersection of $\sol$ with all path gadgets in this bundle, builds a state assignment $\pi : \bund \rightarrow \pstates$. Let $\pbund$ be the path bundle $\bund$ belongs to, and $V_i$ be the variable group corresponding to $\pbund$. Such state assignments $\pi$ correspond to different partial Boolean assignments over $V_i$. Since it holds that
\[
    \big|\pstates^{\bund}\big| = 12^t \geq 12^{t_0 \cdot \log_{12}2} = 2^{t_0} \geq \big|\{\true, \false\}^{V_i}\big|,
\]
we can fix an arbitrary injective mapping $\tau: \{\true,\false\}^{V_i} \rightarrow \pstates^{[t]}$. Since both the set of the variables, and all path gadgets in each bundle are ordered (the former by index, and the latter by the index of the path sequence they belong to), we can assume that the same mapping is defined for all variable groups $V_i$, and all bundles $\bundle{i}{j}$.
Now we describe the different types of gadgets.

\subparagraph*{Simple path gadget $Y$.} Before we describe a path gadget, we describe a simpler structure called a simple path gadget. A simple path gadget consists of the vertices $v, u, w, z$, where $u,w,z$ are all adjacent to $r$, and four more vertices $x_{\conn}, x_{\disc}, x_{\black}, x_{\white}$ called the state vertices, all adjacent to $r$. We color $x_{\black}$ black, and $x_{\white}$ white. We add the triangles $\tri{v}{u}, \tri{u}{w}, \tri{u}{z}, \tri{z}{v}$, and for all $\{s,t\}\in\binom{\States}{2}$ we add the triangle $\tri{x_s}{x_t}$, i.e.\ we add a triangle between each pair of state vertices. We also add the triangles $\tri{u}{x_{\black}}, \tri{u}{x_{\white}}, \tri{w}{x_{\conn}}, \tri{w}{x_{\disc}}$. Finally, we add the edges $\{v, x_{\black}\}, \{v, x_{\white}\}$ and subdivide each of them by a new vertex, and add the edge $\{v, x_{\disc}\}$.
In total, a simple path gadget contains $10$ non-triangle vertices.

Before we define a path gadget, we define the mapping $\nxt \colon \pstates\rightarrow \States^3$ given by:

\begin{center}
    \begin{tabularx}{\textwidth}{X X l}  
    \bp $\nxt(s_{ 1}) = (\black, \black, \white)$&
    \bp $\nxt(s_{ 2}) = (\disc , \black, \white)$&
    \bp $\nxt(s_{ 3}) = (\disc , \disc , \white)$\\
    \bp $\nxt(s_{ 4}) = (\disc , \disc , \black)$&
    \bp $\nxt(s_{ 5}) = (\disc , \disc , \disc )$&
    \bp $\nxt(s_{ 6}) = (\conn , \black, \white)$\\
    \bp $\nxt(s_{ 7}) = (\conn , \disc , \white)$&
    \bp $\nxt(s_{ 8}) = (\conn , \disc , \black)$&
    \bp $\nxt(s_{ 9}) = (\conn , \conn , \disc )$\\
    \bp $\nxt(s_{10}) = (\white, \white, \white)$&
    \bp $\nxt(s_{11}) = (\black, \black, \black)$&
    \bp $\nxt(s_{12}) = (\conn , \conn , \conn )$\\   
    \end{tabularx}\end{center}

\subparagraph*{Path gadgets $\pgadget$.} A path gadget $\pgadget$ consists of six simple path gadgets $Y_1, \dots, Y_6$. We denote a vertex $x$ in $Y_i$ by $x^i$. Hence, by $v^1$ we denote the vertex $v$ in $Y_1$, and by $x_{\conn}^3$ we denote the vertex $x_{\conn}$ in $Y_3$. We call the vertices $v^1, v^2, v^3$ the \emph{entry} vertices, and the vertices $v^4, v^5, v^6$ the \emph{exit} vertices of $X$. These vertices play a special role in $G$, since they connect consecutive path gadgets on a path sequence. In particular, we add a biclique between the exit vertices of each path gadget (except for the last one), and the entry vertices of the following path gadget. The entry vertices of the first path gadget of each path sequence are all adjacent to $g_1$, and the exit vertices of the last path gadget are all adjacent to $g_2$.

For each $s\in \pstates$, we add a vertex $z_s$ to $X$ and make it adjacent to $r$. We call the resulting vertices the \emph{transition vertices} of $X$. We add all triangles $\tri{z_s}{z_{s'}}$ for all $s\neq s' \in \pstates$.
Finally, for each $s \in \pstates$, let $s = (\chi_1,\chi_2,\chi_3)$ and $\nxt(s) = (\chi_4,\chi_5,\chi_6)$, where $\chi_i\in\States$ for all $i\in[6]$. For each $i\in[6]$ we add all triangles $\tri{z_s}{x_{\chi}^i}$, for all $\chi\neq \chi_i\in\States$. That means we add a triangle between a transition vertex $z_s$ and a state vertex $x_{\chi}^i$ in $Y_i$, if $\chi_i\neq \chi$. This completes the description of a path gadget. Given a vertex $v$ in a path gadget $X$, if not clear from the context, we write $v(X)$ to indicate the path gadget it belongs to.
In total, a path gadget contains $6\cdot 10 + 12 = 72$ non-triangle vertices.
% and $6 \cdot 14 + \binom{12}{2} + 6\cdot 12 \cdot 3 = 84 + 66 + 216 = 366$ triangle vertices. This totals in $438$ vertices.

\subparagraph*{Decoding gadgets $\dgadget$.}
Now we describe decoding gadgets, and their connections to the rest of the graph.
As discussed earlier, we attach a decoding gadget $\dgad{i}{j}$ to each bundle $\bundle{i}{j}$. Recall that a bundle is a sequence of $t$ path gadgets. We define a \emph{deletion pair} as two vertices $u, w$, both adjacent to $r$, and a triangle added between them. We define the family of state assignments $\Pi = \pstates^{[t]}$. A decoding gadget consists of $12^t$ disjoint deletion pairs, each corresponding to a mapping in $\Pi$. For $\pi \in \Pi$, let $u_{\pi}, w_{\pi}$ be the deletion pair corresponding to the mapping $\pi$.

Let $\dgadget = \dgad{i}{j}$, and let $\pgadget_1,\dots \pgadget_t$ be the path gadgets on the bundle $\bundle{i}{j}$. For each mapping $\pi\in\Pi$ and for each $i\in[t]$ and $s\in \pstates$, we add the triangle $\tri{z_s(\pgadget_i)}{u_{\pi}}$ if it holds that $\pi(i) \neq s$.
Given a vertex $v$ in a decoding gadget $\dgadget$, if not clear from the context, we write $v(\dgadget)$ to indicate the decoding gadget it belongs to.
In total, a decoding gadget contains  $2 \cdot 12^t$ non-triangle vertices.
%, and $12^t$ triangle vertices, totaling in $3\cdot 12^t$ vertices. In addition, there exist $12^t \cdot t \cdot 11$ triangle vertices between a decoding gadget, and the path gadgets in the bundle it is attached to.

\subparagraph*{Clause gadget \cgadget.}
% column: j
% clause h
% literal group i
% group ell 
A clause gadget $\cgadget = \cgad{j}$ corresponds to the clause $C_{h}$, where $j = k\cdot m + h$ for some $k\in[12\cdot n' - 1]_0$, i.e.\ $\cgadget$ is attached to the $h$th column in its section.
Let $C_h = C^1_h \dot\cup \dots \dot\cup C^{d'}_h$ for some $d'\leq d$ be the partition of $C_h$ into the different variable groups, i.e.\ for all $\ell\in[d']$ it holds that $\var(C^{\ell}_{h})\subseteq V_q$ for some $q \in [s]$, and that for each $q\in[s]$ there exists at most one $\ell\in[d']$ such that $\var(C^{\ell}_{h})\cap V_q \neq \emptyset$.

A clause gadget $\cgad{j}$ consists of a single odd cycle over $d'$ vertices, if $d'$ is odd, or $d'+1$ vertices otherwise. Let $c_1, \dots c_{d'}$ be the first $d'$ vertices on this cycle, and $c_0$ be the extra vertex if $d'$ is even.

Let $V^1_h, \dots, V^{d'}_h$ be the variable groups underlying $C^1_h, \dots, C^{d'}_h$ respectively. For $i\in[d']$, let $V^i_h = V_{\ell}$ be the $\ell$th variable group.
We add edges between $c_i$ and $\dgad{\ell}{j}$ as follows:

Let $A$ be the set of all partial assignments over $V_{\ell}$ satisfying $C^{i}_h$, and let $\hat{\Pi} = \tau(A)$ be the set of all corresponding state assignments. Then we add all edges $\{c_i, w_{\pi}\}$ for each $\pi \in \hat{\Pi}$. This completes the description of a clause gadget, and with which the description of the graph $G$.

\subparagraph*{Budget.}
Now we define the target value $\budget$ in the resulting instance $(G, \budget)$ of the \Coctp problem. First, we define a family of disjoint components $\lbfamily$ in $G$ (mostly the gadgets of $G$). For each component $X\in\lbfamily$, we will define a family of disjoint subsets of $X$ (denoted $\lbsets(X)$). For each set $S \in \lbsets(X)$, we will define a lower bound $\lbound(S)$ on the size of a connected odd cycle transversal in $G[S]$. We define $\lbound(X)$ as the sum of $\lbound(S)$ for all $S\in\lbsets(X)$. Finally, we define $\budget$ as the sum of $\lbound(X)$ for all $X\in \lbfamily$.

By the tightness of $\budget$, as we shall see in \cref{lem:lb-b-tight}, it holds that any solution of size $\budget$ must contain exactly $\lbound(X)$ vertices in each such component $X$, and no other vertices.

Let $\pgfamily, \dgfamily$ and $\cgfamily$ be the families of all path gadgets, decoding gadgets and clause gadgets in $G$ respectively. Let
\[\lbfamily = \pgfamily \cup \dgfamily\cup \cgfamily \cup \big\{\inctri{\{r\}}\big\} \cup \big\{\inctri{\{g_1\}}\big\} \cup \big\{\inctri{\{g_2\}}\big\}\]
For a path gadget $\pgadget$ we define $\lbsets(\pgadget)$ as the family of all sets:
\begin{itemize}
    \item $\inctri{\{x^i_{\conn}, x^i_{\disc}, x^i_{\black}, x^i_{\white}\}}$, for all $i\in[6]$.
    \item $\inctri{\{v^i, u^i, w^i, z^i\}}$, for all $i\in[6]$.
    \item $\inctri{\{z_s \colon s \in \pstates\}}$.
\end{itemize}
We define
\begin{itemize}
    \item $\lbound(\inctri{\{x^i_{\conn}, x^i_{\disc}, x^i_{\black}, x^i_{\white}\}})=3$, for all $i\in[6]$.
    \item $\lbound(\inctri{\{v^i, u^i, w^i, z^i\}})=2$, for all $i\in[6]$.
    \item $\lbound(\inctri{\{z_s \colon s \in \pstates\}})=11$.
\end{itemize}
It follows that $\lbound(\pgadget) = 41$.
For a decoding gadget $\dgadget$, we define
\[\lbsets(\dgadget) = \big\{\inctri{\{u_{\pi}, w_{\pi}\}}\colon \pi \in \Pi\big\}.\]
For all $\pi \in \Pi$ we set 
\[\lbound(\inctri{\{u_{\pi}, w_{\pi}\}}) = 1.\]
It follows that $\lbound(\dgadget) = 12^t$.

For a clause gadget $\cgadget$, we define $\beta(\cgadget) = \{\cgadget\}$, and set $\lbound(\cgadget) = 1$. Similarly, for $x\in\{r,g_1,g_2\}$, we define $\lbsets(\inctri{\{x\}}) = \big\{\inctri{\{x\}}\big\}$, and set $\lbound(\inctri{\{x\}}) = 1$

It follows that 
\[
\budget = 41 \cdot c  \cdot n' + 12^t \cdot c \cdot s + c + 3
        = (41 \cdot n' + 12^t \cdot s + 1)\cdot c + 3.
\]

\subsection{Proof of correctness}

\begin{lemma}\label{lem:lb-cw-bound}
    The graph $G$ has linear clique-width at most $k = n' + k_0$, where $k_0$, is a constant upper bound for the total number of
    non-triangle vertices in a path gadget, a decoding gadget, and a clause gadget plus seven, i.e.\
    \[
        k_0 = 72 + 2\cdot 12^t + d + 1 + 7 = d + 2 \cdot 12^t + 80.
    \]
    Moreover, a linear $k$-expression of $G$ can be constructed in polynomial time.
\end{lemma}

\begin{proof}
    We describe a linear $k$-expression of $G$. Let $x_0, x_1, \dots$ be the nodes of $\syntaxtree$ building a simple path, and $x'_i$ be the private neighbor of $x_i$ (if exists) in the caterpillar structure of the syntax tree of a linear $k$-expression. Since we only use these nodes to introduce a new vertex of some label $\ell$, and then unify it at $x$ with the labeled graph built so far $G_{x_{i-1}}$, we will just ignore these nodes and say that at $x_i$ we add a vertex labeled $\ell$ to $G_{x_{i-1}}$. Hence, we consider $x_0, x_1, \dots$ as timestamps of building $G$, and denote $G_{x_i}$ by $G_i$. Similarly, let $\mu_i = \mu_{x_i}$, $V_i = V_{x_i}$, $E_i = E_{x_i}$, and $\lab_i = \lab_{x_i}$.

    We reserve the first $n'$ labels for later, and use the label $\lforget = n'+1$ to label vertices whose incident edges have all been introduced before the current timestamp. 
    When we say we forget a vertex $u$ at some timestamp $x_i$, we mean that we relabel this vertex to $\lforget$, i.e.\
    \[\mu_i = \relabel{\lab_{i-1}(u)}{\lforget}(\mu_{i-1}).\]
    We also reserve the label $\ltri = n'+2$ for adding triangles between two vertices. Let $x$ be some timestamp. For $u,v\in V_x$, let $i=\lab_x(u)$ and $j=\lab_x(v)$. When say that we add a triangle between $u$ and $v$, we assume first that $|\rlab_x(i)| = \{u\}$ and $|\rlab_x(j)| = \{v\}$ (both have unique labels). In this case, we mean that we add a new vertex $w_{u,v}$ with label $\ltri$, and add all edges between $u$, $v$ and $w_{u,v}$, and then we relabel $w_{u,v}$ to $\lforget$. Since this is the only time we use the label $\ltri$ in the construction, clearly at any given timestamp it holds that $|\rlab(\ltri)|\leq 1$.

    Now we are ready to describe the linear $k$-expression of $G$. First, we introduce the vertices $r,\broot, g_1$ and $g_2$ with labels $\lroot = n'+3$, $\lbroot=n'+4$, $\lgl = n'+5$, and $\lgr=n'+6$ respectively. These are the only vertices that will be assigned the corresponding labels at any timestamp, and will never be relabeled. Hence, for a vertex $u$ of a unique label at any timestamp, we can make it adjacent to any of them through a join operation without changing the adjacency of any other vertex of the graph. We introduce the private neighbors of $r, g_1,$ and $g_2$, using a different label for each introduced vertex. We create the corresponding triangles using join operations, and then we forget all these private neighbors.
    Finally, we reserve the label $\lwhite = n'+7$ for creating white vertices when coloring some vertices black. For a vertex $v$ of a unique label $\ell$, when we say we color $v$ white we mean that we join the labels $\ell$ and $\lbroot$, while by coloring $v$ black, we mean that we introduce a vertex $w$, labeled $\lwhite$, and then apply two join operations: one between $\lwhite$ and $\lbroot$, and another between $\lwhite$ and $\ell$. After that we forget the label $\lwhite$.

    We divide the rest of the construction into phases. In each phase we add a new column to the construction, one by one, adding the $j$th column in the phase $j$. Let $x_{h}$ be the timestamp at the end of phase $j$. We will preserve the following invariant: 
    It holds that $\lab_{h}(r) = \lroot$, $\lab_{h}(\broot) = \lbroot$, $\lab_{h}(g_1) = \lgl$, $\lab_{h}(g_2) = \lgr$, and 
    \[\lab_{h}(v^4(X_i^j)) = \lab_{h}(v^5(X_i^j)) = \lab_{h}(v^6(X_i^j)) = i\] for all $i\in [n']$.
    All other vertices in $V_{h}$ are labeled $\lforget$ at $x_h$.

    Now we describe the phase $j$ for all $j\in [c]$. The phase starts by creating the vertices of the clause gadget $\cgad{j}$, each having its own label, and we add the edges of the odd cycle accordingly. After that we build the bundles one by one. For a bundle $\bundle{\ell}{j}$, we first create the decoding gadget $\dgad{\ell}{j}$ by creating all non-triangle vertices of the deletion pairs, assigning a unique label to each vertex, and adding a triangle between the vertices of each deletion pair. After that we add the edges between $\cgad{j}$ and $\dgad{\ell}{j}$, and the edges between $r$ and $\dgad{\ell}{j}$ as described in the construction of $G$. Finally, we build the path gadgets of $\bundle{\ell}{j}$ one by one, where for a path gadget $\pgad{i}{j}$, we create all non-triangle vertices of $\pgad{i}{j}$, each of a unique label, we add all edges and triangles between the vertices of $\pgad{i}{j}$, and between $\pgad{i}{j}$ and $\bundle{\ell}{j}$. We color $x^t_{\black}(\pgad{i}{j})$ black, and $x^t_{\white}(\pgad{i}{j})$ white, for all $t\in[6]$. If $j = 1$, we add all edges between the entry vertices of $\pgad{i}{j}$ and $g_1$. Otherwise, we add the edges between the entry vertices of $\pgad{i}{j}$ and the exit vertices of $\pgad{i}{j-1}$ by joining their labels with the label $i$. If $j = c$, we add all edges between the exit vertices of $\pgad{i}{j}$ and $g_2$.

   After we build a path gadget, we forget the exit vertices of $\pgad{i}{j-1}$ if $j\neq 1$. We also forget all vertices of $\pgad{i}{j}$ except its exit vertices, and relabel its exit vertices to the label $i$. After we build the whole bundle we forget its decoding gadget, and after we build the whole column we forget the clause gadget attached to it.
\end{proof}

Now we prove the correctness of the reduction.

\begin{lemma} \label{lem:lb-sol-if-sat}
    If $I$ is satisfiable, then $G$ admits a connected odd cycle transversal of size $\budget$.
\end{lemma}

\begin{proof}
    Let $\alpha$ be a satisfying assignment of $I$. For $\ell \in [s]$, let $\alpha_{\ell}$ be the restriction of $I$ to $V_{\ell}$, and let $\pi_{\ell} = \tau(\alpha_{\ell})$. We define the solution $\sol \subseteq V$ of size $\budget$ as the set of the following vertices:

    \begin{itemize}
        \item The root $r$, and the guards $g_1$ and $g_2$.
        
        \item For each $\ell \in [s]$, we include in $\sol$ the same vertices from all bundles $\bundle{\ell}{j}$ on the path bundle $\pbundle{\ell}$, hence $\sol$ will contain the same vertices from all path gadgets on the same path sequence. Let us fix some bundle $\bund$ on this path bundle, and let $\pgadget_1, \dots \pgadget_t$ be the path gadgets along this bundle, i.e.\ $\pgadget_i = \pgad{i'}{j}$ for $i' = (\ell - 1) \cdot t + i$. We describe which vertices to include from this path gadget.
        
        Given a simple path gadget $Y$, and a state $s \in \States$,
        we define $L_{\chi}(Y) = \{v, w\}$ if $\chi \in \{\conn, \disc\}$, and $L_{\chi}(Y) = \{u, z\}$ otherwise. We define $R_{\chi} = \{x_{{\chi}'} \colon {\chi}' \neq {\chi}\}$. Let $V_{\chi}(Y) = L_{\chi}(Y)\cup R_{\chi}(Y)$.
        Let $\pi_{\ell}(i) = s = (\chi_1, \chi_2, \chi_3)$, and let $\nxt(s) = (\chi_4, \chi_5, \chi_6)$.
        We include in $\sol$ all vertices $V_{\chi_i}(Y_i)$ for all $i\in[6]$. We also include all vertices $z_{s'}$ for $s' \neq s$.
        From $\dgad{\ell}{j}$, we include all vertices $u_{\pi'}$ for $\pi' \neq \pi_{\ell}$, and the vertex $w_{\pi_{\ell}}$ in $\sol$.

        \item Let $\cgad{j}$ be a clause gadget, and assume that $j= k\cdot m  + h$ for some value of $k$, i.e.\ $\col{j}$ is the $h$th column in its section. Let $i \in [d']$ be the smallest index, such that $\alpha$ satisfies $C^i_h$. We include the vertex $c_i$ in $\sol$.
    \end{itemize}

    We claim that $\sol$ is a connected odd cycle transversal of $G$ of size $\budget$. The size follows from the fact that we pack exactly $\lbound(Q)$ vertices in each component $Q \in \lbsets(X)$ for all $X \in \lbfamily$. First we show that $G[\sol]$ is connected, and after that we provide a $2$-coloring of $G\setminus \sol$. Since it holds that $r \in \sol$, in order to show connectivity it suffices to show that all vertices in $\sol$ are reachable from $r$ in $G[\sol]$.
    
    By the choice of $\sol$, it does not contain any triangle vertex. It holds the $r$ is adjacent to $g_1, g_2$, to both endpoints of any deletion pair, and to all non-triangle vertices in a path gadget except for $v^1,\dots v^6$. Hence, it suffices to show that entry vertices, exit vertices and vertices of clause gadgets that are in $\sol$ are reachable from $r$ in $G[\sol]$. For a clause gadget $\cgad{j}$, let $C_{h} = \{a^{h}_1, \dots a^{h}_d\}$ be the corresponding clause (i.e.\ $\col{j}$ is the $h$th in its section). By the choice of $\sol$, it holds that there exists a unique $i\in[d']$ with $c_i \in \sol$. It holds by the choice of $i$, that $\alpha$ satisfies $C^i_h$. Let $V_{\ell} \supseteq \var(C^i_h)$ be the corresponding variable group. Then it holds that $\alpha_{\ell}$ satisfies $C^i_h$, and hence, $c_i$ is adjacent to $w_{\pi_{\ell}}(\dgad{\ell}{j})$, which is in $\sol$ and is adjacent to $r$.
    
    Now let us fix a bundle $\bundle{\ell}{j}$, and let $\pgadget$ be the $i$th path gadget on this bundle. For $h\in[3]$, let $v^h \in \sol$ be an entry vertex in the solution.
    If $\pgadget$ is the first path gadget on its path sequence, then $v^h$ is adjacent to $g_1$, which is in $\sol$ and adjacent to $r$. Otherwise, let $\pgadget'$ be the path gadget that precedes $\pgadget$ on its path sequence.
    Let $\pi_{\ell}(i) = s = (\chi_1,\chi_2,\chi_3)$, and $\nxt(s) = (\chi_4, \chi_5, \chi_6)$. Then it holds that $\chi_h \in \{\disc, \conn\}$. If $\chi_h = \conn$, then it holds that $x^h_{\disc}(\pgadget) \in \sol$, and is adjacent to $v^h$. Otherwise, it holds that $\chi_h = \disc$. From the definition of the mapping $\nxt$, it holds that $\conn \in \{\chi_4, \chi_5, \chi_6\}$. Hence, there exists $h'\in \{4,5,6\}$, such that both $x^{h'}_{\disc}(\pgadget')$ and $v^{h'}(\pgadget')$ are in $\sol$. Hence, we get the path $r, x^{h'}_{\disc}(\pgadget'), v^{h'}(\pgadget'), v^h$, and $v^h$ is reachable from $r$.

    Analogously, for $h\in \{4,5,6\}$, let $v^h \in \sol$ be an exit vertex in the solution.
    If $\pgadget$ is the last path gadget on its path sequence, then $v^h$ is adjacent to $g_2$, which is in $\sol$ and adjacent to $r$. Otherwise, let $\pgadget'$ be the path gadget that follows $\pgadget$ on its path sequence.
    Let $\pi_{\ell}(i) = s = (\chi_1,\chi_2,\chi_3)$, and $\nxt(s) = (\chi_4, \chi_5, \chi_6)$. Then it holds that $\chi_h \in \{\disc, \conn\}$. If $\chi_h = \conn$, then it holds that $x^h_{\disc}(\pgadget) \in \sol$, and is adjacent to both $v^h$ and $r$. Otherwise, it holds that $\chi_h = \disc$. From the definition of the mapping $\nxt$, it holds that $\conn \in \{\chi_1, \chi_2, \chi_3\}$. Hence, there exists $h'\in \{1,2,3\}$, such that both $x^{h'}_{\disc}(\pgadget')$ and $v^{h'}(\pgadget')$ are in $\sol$. Hence, we get the path $r, x^{h'}_{\disc}(\pgadget'), v^{h'}(\pgadget'), v^h$, and $v^h$ is reachable from $r$.

    Now we show that $G\setminus \sol$ is $2$-colorable by presenting a proper 2-coloring with colors $\black$ and $\white$. First note that each triangle vertex has degree one in $G \setminus \sol$. Hence, we can always properly color these vertices assuming the rest of the graph admits a proper $2$-coloring. Let $G'$ be the graph resulting from $G\setminus \sol$ after removing all triangle vertices. We show that $G'$ is $2$-colorable.

    For a decoding gadget $\dgadget$, it holds for $\pi \in \pstates^{[t]}$ that either $w_{\pi} \in \sol$, or $u_{\pi}\in\sol$. In the former case, it holds that $u_{\pi}$ is an isolated vertex in $G'$, since all its adjacencies were built through triangles, and hence, all its non-triangle neighbors belong to $\sol$. In the latter case, it holds that $w_{\pi}$ is either isolated in $G'$, or adjacent to at most one vertex of a clause gadget, namely the one corresponding to the same variable group this decoding gadget corresponds to. Since decoding gadgets induce simple paths in $G'$, it holds that all decoding gadgets together with all clause gadgets build connected components of $G'$ that are either isolated vertices or caterpillar graphs. Both are trivially $2$-colorable.
    
    Finally, We are left with the vertices of path gadgets, together with the vertex $\broot$.
    We show that we can define a proper 2-coloring $\gamma$ of these vertices that extends the coloring $\{\broot  \mapsto \black\}$. First, we assign the color $\black$ to all vertices that were colored black, and the color $\white$ to all vertices colored white in the construction of $G$. Since $\broot$ separates the different path sequences in $G'$, it suffices to fix $i\in [n']$ and show that we can color all vertices of the path gadgets on the $i$th path sequence consistently with the fixed coloring of $\broot$. We color all path gadgets $\pgad{i}{j}$ in the same way for all values of $j$. Recall that $\sol$ includes the same vertices in all path gadgets on the $i$th path sequence. Assume that $\pgad{i}{j}$ is the $h$th path gadget on the $\ell$th bundle, and let $\pi_{\ell}(h) = s = (\chi_1, \chi_2, \chi_3)$. Let $\nxt(s) = (\chi_4,\chi_5,\chi_6)$. Fix $q \in [6]$, let $\chi = \chi_q$. We show how to color the simple path gadget $Y = Y_q$.  If $\chi \in \{\black, \white\}$, let $\chi' = \{\black,\white\}\setminus \chi$. We assign the color $\chi$ to $v$. The vertex $w$ is isolated, and hence, can be colored arbitrarily. Finally, we assign the color $\chi'$ to the two vertices added to subdivide the edges $\{v, x_{\black}\}$ and $\{v, x_{\white}\}$. On the other hand, if $\chi \in \{\conn, \disc\}$, then the vertices $u, z$ and the vertices added to subdivide edges are all isolated, and can be colored arbitrary.
    
    This already colors all simple path gadgets $Y_1,\dots Y_6$. Finally, the vertices $z_s$ are only connected with triangles to other vertices in $G$, and hence, are isolated in $G'$ and can be colored arbitrarily. It is easy to see that the described coloring induces a proper 2-coloring in each path gadget. So we only have to check the edges between the exit vertices of a path gadget and the entry vertices of the following path gadget. For $\chi \in \{\black, \white\}$, if $\chi \in \{\chi_3,\chi_4,\chi_5\}$, then $\chi \notin \{\chi_1, \chi_2, \chi_3\}$. Hence, if $v^h(\pgad{i}{j})$ is assigned $\chi$ for some $h\in\{3,4,5\}$, then $v^{h'}(\pgad{i+1}{j})$ is not assigned $\chi$ for all $h'\in\{1,2,3\}$.
\end{proof}

Finally, we prove the other direction of the reduction. In order to achieve this, we make some observations on the structure of an optimal solution.
From now on, let $\sol$ be a given connected odd cycle transversal of $G$ of size $\budget$.

\begin{lemma}\label{lem:lb-b-tight}
    For each $X\in \lbfamily$, and $Q \in \lbsets(X)$, let $S_Q = \sol\cap Q$. Then it must hold that $|S_Q| = \lbound(Q)$. Moreover, it holds that
    \[
        S = \dot{\bigcup\limits_{X\in\lbfamily}}
        \left(\dot{\bigcup\limits_{Q\in\lbsets(X)}}S_Q\right).
    \]
\end{lemma}

\begin{proof}
    We show that $|S_Q| \geq \lbound(Q)$ must hold. The equality and the fact that $S_Q$ does not contain any other vertices follows from the tightness of $\budget$, since all components $Q$ are pairwise disjoint.

    If $Q = \inctri{\{x\}}$ for $x\in\{r,g_1,g_2\}$, $Q=\cgadget$ for a clause gadget $\cgadget$, or $Q=\inctri{\{u_{\pi}, w_{\pi}\}}$ for a deletion pair in a decoding gadget, the claim follows from the fact, that $Q$ induces an odd cycle in $G$.
    For a path gadget $\pgadget$, we distinguish the different types of the set $Q$.

    \begin{itemize}
        \item Let $Q = \inctri{\{x^i_{\conn}, x^i_{\disc}, x^i_{\black}, x^i_{\white}\}}$ for some $i\in[6]$. Since there is a triangle between any two vertices $x^i_{\chi}$ and $x^i_{\chi'}$ for $\chi \neq \chi'$, we need to include at least three of these four vertices, or their private neighbors in order to hit all triangles. Hence, it must hold that $|S_Q\geq 3|$.
        
        \item Let $Q = \inctri{\{v^i, u^i, w^i, z^i\}}$. Let $S_1 = \{v^i, w^i\}$, and $S_2 = \{u^i, z^i\}$. If $S_1\not \subseteq S_Q$, then it must hold that $S_2 \subset S_Q$ since we added all triangles between each vertex of $S_1$ and each vertex of $S_2$.
        
        \item For $Q = \inctri{\{z_s \colon s \in \pstates\}}$, similarly to the first case, there is a triangle between any pair of vertices in $\{z_s \colon s \in \pstates\}$. Hence, we need to include at least $12-1 = 11$ of these vertices (or their private neighbors) in $S$ in order to hit all triangles. Hence, it holds that $S_Q\geq 11$.
    \end{itemize}
    
\end{proof}

\begin{lemma}\label{lem:lb-sol-no-tri}
    Any solution $\sol$ of size $\budget$ cannot contain any triangle vertices.
\end{lemma}

\begin{proof}
    Since $G$ contains at least three vertex disjoint odd cycles (the triangles at $r, g_1$ and $g_2$ for instance), it must holds that $|\sol| \geq 3$. Let $w$ be a triangle vertex in $\sol$. Since $G[\sol]$ is connected, and since the non-triangle neighbors of $w$ separates $w$ from the rest of the graph, there must exist a neighbor $u$ of $w$ that belongs to $\sol$ and is not a triangle vertex. But each odd cycle containing $w$ must containing $u$ as well, since $w$ has degree two in $G$. We claim that $\sol' = \sol \setminus \{w\}$ is a connected odd cycle transversal of $G$ of size $\budget -1$ which contradicts \cref{lem:lb-b-tight}.
    Since the neighbors of $w$ are a subset of the neighbors of $u$, $G[\sol']$ is connected. Since $w$ has degree at most $1$ in $G\setminus \sol'$, we can extend any proper 2-coloring of $G\setminus S$ to a proper 2-coloring of $G\setminus \sol'$ by assigning a different color to $w$ than its non-solution neighbor (if exists). 
\end{proof}

\begin{corollary}\label{lem:lb-tri-imply-sol}
    It holds for each vertex $v$ with a triangle added at $v$ that $v\in \sol$, and for each pair $u,v$ with a triangle added between $u$ and $v$ that $\{u,v\}\cap S\neq \emptyset$.
\end{corollary}

\begin{proof}
    This follows from \cref{lem:lb-sol-no-tri}, since $v$ in the former case, and $u$ and $v$ in the latter, are the only non-triangle vertices on odd cycles of length three.
\end{proof}

\begin{definition}\label{def:lb-pgad-state}
    Let $\pgadget$ be some path gadget. Let $z_s$ be the only transition vertex of $\pgadget$ not in $\sol$ for some $s \in \pstates$. Then we call $s$ the \emph{state} of $\pgadget$ defined by $\sol$ (denoted by $\statef_{\sol}(\pgadget)$). Similarly, for a simple path gadget $Y$, let $x_{\chi}$ be the only state vertex of $Y$ not in $\sol$ for some $\chi\in\States$. Then we call $\chi$ the \emph{state} of $Y$ defined by $\sol$ ($\statef_{\sol}(Y)$). We omit the subscript $\sol$ when clear from context.
\end{definition}

\begin{lemma}\label{lem:lb-state-implies-simple-state}
    Let $\pgadget$ be some path gadget. Let $s = (\chi_1,\chi_2,\chi_3)$ be the state of $\pgadget$, and let $\nxt(s) = (\chi_4,\chi_5,\chi_6)$. Then, it must hold that $\chi_i$ is the state of $Y_i$ for all $i\in [6]$.
\end{lemma}

\begin{proof}
    Assume this is not the case. Let $i\in [6]$, with $x^i_{\chi} \notin \sol$ for some $\chi\neq \chi_i$. But then it holds that both $x^i_{\chi}$ and $z_s$ are not in $\sol$ which contradicts \cref{lem:lb-tri-imply-sol} since we added a triangle between these two vertices.
\end{proof}

\begin{lemma}\label{lem:lb-state-imply-v}
    Let $Y$ be a simple path gadget. We define $V_Y^{\sol} = (V(Y)\cup\{r\})\cap \sol$. Let $\chi = \statef(Y)$. Then one of the following is true:
    \begin{itemize}
        \item $\chi = \conn$, $v \in \sol$, and $v$ is reachable from $r$ in $G[V_Y^{\sol}]$,
        
        \item $\chi = \disc$, $v \in \sol$, and $v$ is isolated in $G[V_Y^{\sol}]$,
        
        \item $\chi = \black$, $v \notin \sol$, and $\gamma(v) = \black$ for each proper 2-coloring of $G\setminus \sol$ with colors $\black$ and $\white$ that extends $\broot\mapsto\black$,
        
        \item or $\chi = \white$, $v \notin \sol$, and $\gamma(v) = \white$ for each proper 2-coloring of $G\setminus \sol$ with colors $\black$ and $\white$ that extends $\broot\mapsto\black$.
    \end{itemize}
\end{lemma}

\begin{proof}
    Since $\budget(\{v,u,w,z\}) = 2$, it must hold that either $\{v, w\} \subseteq \sol$ and $\{u, z\}\cap \sol = \emptyset$, or $\{u, z\}\subseteq \sol$ and $\{v, w\}\cap \sol = \emptyset$. If $\chi \in \{\conn, \disc\}$, then the former must be the case due to the triangle $\tri{w}{x_{\chi}}$, and otherwise  the latter must be the case due to the triangle $\tri{u}{x_{\chi}}$.

    If $\chi = \conn$, then $v$ is reachable from $r$ through the vertex $x_{\disc}$. On the other hand, if $\chi = \disc$, then $v$ is isolated in $G[V_Y^{\sol}]$.

    Finally, assume that $\chi = \black$, or $\chi = \white$, and that $v\notin \sol$.
    Let $\gamma$ be an arbitrary proper $2$-coloring of $G\setminus \sol$ that extends $\broot \mapsto \black$. Then there is a fixed path from $\broot$ to $v$ in $G\setminus S$ that goes from $\broot$ to $x_{\chi}$ (in the former case through an added white vertex), and then to $v$ through the subdivision vertex between $v$ and $x_{\chi}$. In the former case the path has an even length and hence, it must hold that $\gamma(v) = \black$, while in the latter the path has an odd length and hence, $\gamma(v) = \white$ must hold.
\end{proof}

\begin{lemma}\label{lem:lb-states-transition-direction}
    For each two consecutive path gadgets $\pgadget$ and $\pgadget'$, let $s_i$ be the state of $\pgadget$ and $s_{i'}$ be the state of $\pgadget'$. Then it holds that $i'\leq i$.
\end{lemma}

\begin{proof}
    Assume this is not the case. Then there exists two consecutive path gadgets $\pgadget$ and $\pgadget'$ with states $s_i$ and $s_{i'}$ respectively, such that $i' > i$.
    Let $\nxt(s_i) = (\chi_1, \chi_2, \chi_3)$, let $s_{i'} = (\chi'_1, \chi'_2, \chi'_3)$.
    Let $v_1 = v^4(\pgadget)$, $v_2 = v^5(\pgadget)$, $v_3 = v^6(\pgadget)$ be the exit vertices of $\pgadget$, and 
    Let $v'_1 = v^1(\pgadget')$, $v'_2 = v^2(\pgadget')$, $v_3 = v'^3(\pgadget')$ be the entry vertices of $\pgadget'$.
    We will be interested in the sets $U = \{v_1, v_2, v_3\}$, and $U' = \{v'_1, v'_2, v'_3\}$. 
    We define the sets $L = \{\chi_1, \chi_2, \chi_3\}$ and $R = \{\chi'_1, \chi'_2, \chi'_3\}$. 
    
    From the construction of $G$, it holds that $(U, U')$ induces a complete bipartite subgraph of $G$ with bipartition $(U,U')$.
    Hence, the following claims follows from \cref{lem:lb-state-implies-simple-state} and \cref{lem:lb-state-imply-v}.
    \begin{itemize}
        \item It must hold that $(L \cap R) \cap \{\black, \white\} = \emptyset$.
        \item If $\disc \in  L$, then it holds that $\conn \in R$ or that both $\disc \in R$ and $\conn \in L$.
        \item If $\disc \in  R$, then it holds that $\conn \in L$ or that both $\disc \in L$ and $\conn \in R$.
    \end{itemize}
    The first claim follows from the assumption that $G \setminus \sol$ admits a proper 2-coloring $\gamma$ with colors $\black$ and $\white$. Having the same color $\black$ or $\white$ in both sets $L$ and $R$, implies by \cref{lem:lb-state-imply-v} that $\gamma$ assigns the same colors to a vertex in $U$ and a vertex in $U'$, which contradicts the assumption that $\gamma$ is a proper 2-coloring, since $(U, U')$ induces a complete bipartite graph in $G$.
    
    For the second claim, let $\chi_i = \disc$ for some $i\in[3]$. \cref{lem:lb-state-imply-v} implies that $v_i \in \sol$, but not adjacent to any other vertex in its path gadget in $G[\sol]$. Since $G[\sol]$ is connected, it must hold that there is a vertex $v'_j \in \sol$ for some $j\in[3]$, which implies that $\{\conn, \disc\} \cap R \neq \emptyset$. However, if $\conn \notin R$, then it must hold that $\conn \in L$, since the set $\{v_i \colon \chi_i = \disc\} \cup \{v'_i \colon \chi'_i = \disc \}$ is not empty, and would otherwise build a connected component of $G[\sol]$. The last claim is symmetric to the second, and follows analogously.

    By iterating over the list of states, it is not hard to see that whenever $i' \geq i$ holds, one of these three claims is contradicted.
\end{proof}

\begin{lemma}\label{lem:lb-sat-if-sol}
    If $G$ admits a connected odd cycle transversal of size $\budget$, then $I$ is satisfiable.
\end{lemma}

\begin{proof}
    Let $\sol$ be a solution of size $\budget$. Let us fix some path sequence, and consider the path gadgets $\pgad{i}{1} \dots \pgad{i}{c}$ on this path sequence. Let $s_{i_1} \dots s_{i_c}$ be the states of these path gadgets respectively. Then by \cref{lem:lb-states-transition-direction} it holds that $i_j \leq i_{j+1}$ for all $j\in[c]$. Since we have exactly 12 states, there can exist at most 11 different indices $j$ with $i_j \neq i_{j+1}$. Since we have $n'$ path sequences, we can have at most $11\cdot n'$ columns, where some path gadget on each of these columns admitting a different state than the following path gadget. Hence, at most $11\cdot n'$ sections exist, with some path gadget in each of these sections having a different state than the following path gadget. Since there exist $11\cdot n' + 1$ sections in $G$, there exists at least one section in $G$, where each path gadget in this section admits the same state as the following path gadget.
    
    Let us fix one such section $\sect{r}$ for $r \in [11\cdot n'+1]$. For each segment $\seg{i}{r}$ in this section ($i\in [n']$), let $s_i$ be the state of all path gadgets on this segment.
    Let us group these states into groups, each of size $t$, and define the mappings $\pi_1, \dots \pi_s \in [12]^{[t]}$ where $\pi_{\ell}$ corresponds to the $\ell$th group of states. For $i\in[t]$, we define $\pi_{\ell}(i)$ to be the $i$th state in the $\ell$th group, i.e.\
    \[
        \pi_{\ell}(i) = s_{(\ell-1)t + i}.
    \]
    Finally, we define the Boolean assignments $\alpha_1,\dots \alpha_s$, where $\alpha_{\ell}$ is defined over $V_{\ell}$. We define $\alpha_{\ell} = \tau^{-1}(\pi_{\ell})$, if such an assignment exists, or set $\alpha_{\ell}$ to a fixed arbitrary assignment otherwise.

    We claim that $\alpha_1\cup \dots\cup \alpha_s$ is a satisfying assignment of $I$. In order to see this, let $C_j$ be some clause for $j\in [m]$. Let us consider the $j$th column of the $r$th section. Since $\sol$ is a solution, there must exist a vertex $c$ of the clause gadget attached to his column in the solution. Since $G[\sol]$ is connected, $c$ must have a neighbor $w_{\pi}$ on the $\ell$th decoding gadget in the solution as well. Since this implies that $u_{\pi}$ is not in the solution, it must hold that $\pi = \pi_{\ell}$, and from the construction of $G$, it must hold that $\alpha_{\ell}$ satisfies $C_j$.
\end{proof}

\begin{proof}[Proof of \cref{theo:lower-bound}]
    Assume that the \Coctp\ problem can be solved in time $\ostar((12-\varepsilon)^k)$ given a $k$-expression $\mu$ of $G$. We show that there exists a positive constant $\delta$, such that the $d$-SAT problem can be solved in time $\ostar((2-\delta)^n)$ for all values of $d$, which contradicts SETH. Given an instance $I$ of the $d$-SAT problem, first we build the graph $G$, described above, and then we run the given algorithm on the instance $(G, \budget)$.
    It holds from \cref{lem:lb-sol-if-sat} and \cref{lem:lb-sat-if-sol} that $G$ admits a connected odd cycle transversal of size $\budget$ if and only if $I$ is satisfiable.
    
    Now we show that this algorithm runs in time $\ostar((2-\delta)^n)$ for some positive value $\delta$. The graph $G$ together with a $k$-expression of $G$, for $k = n' + k_0$ for some constant $k_0$, can be constructed in polynomial time by \cref{lem:lb-cw-bound}. Hence, the given algorithm runs in time 
    \[
        \ostar((12-\varepsilon)^k) = \ostar((12-\varepsilon)^{n'}).    
    \]
    It holds that 
    \begin{align*}
        n' &= s\cdot t\\
           &= \left\lceil\frac{n}{t_0}\right\rceil \cdot t\\
           &\leq \left(\frac{n}{t_0} + 1\right) \cdot t\\
           &= \frac{n}{t_0} \left\lceil t_0 \log_{12}{2}\right\rceil + t\\
           &\leq \frac{n}{t_0} \left( t_0 \log_{12}{2} + 1\right) + t\\
           &= n\left(\log_{12}{2} + 1/t_0\right) + t.\\
    \end{align*}
    Hence, the algorithm runs in time 
    \[
    \ostar\left( (12-\varepsilon)^{n\left(\log_{12}{2} + 1/t_0\right) + t}\right)
    = \ostar\left( (12-\varepsilon)^{n\left(\log_{12}{2} + 1/t_0\right)}\right).
    \]
    We claim that one can choose $t_0$ in such a way that
    \[(12-\varepsilon)^{\log_{12} 2 + 1/t_0} < 2,\]
    which implies that there exists a positive $\delta$ with
    \[(12-\varepsilon)^{\log_{12} 2 + 1/t_0} \leq 2 - \delta.\]
    This implies that the algorithm runs in time
    \[
        \ostar\left((2-\delta)^n\right),
    \]
    which contradicts SETH.
    Hence, our goal is to find a positive constant integer $t_0$ such that 
            \[(12-\varepsilon)^{\log_{12} 2 + 1/t_0} < 2,\]
            or in other words
        \[\log_{12} (12-\varepsilon) \cdot (\log_{12} 2 + 1/t_0) < \log_{12}2.\]
    Let $\alpha = \log_{12}(12-\varepsilon)$. We choose $t_0$ to be any constant integer larger than $\frac{\alpha}{\log_{12} 2 \cdot (1-\alpha)}$.
    Then it holds that
    \begin{align*}
        \alpha \cdot (\log_{12} 2 + 1/t_0)
        &< \alpha  \left(\log_{12} 2 + \frac{\log_{12} 2 \cdot (1-\alpha)}{\alpha}\right)\\
        &= \alpha  \left(\frac{\alpha \log_{12} 2}{\alpha} + \frac{\log_{12} 2 \cdot (1-\alpha)}{\alpha}\right)\\
        &= \frac{\alpha \cdot \log_{12}2}{\alpha} = \log_{12}2.
    \end{align*}

\end{proof}

\section{Conclusion}\label{sec:conclusion}
In this work, we have provided a single-sided error Monte Carlo algorithm for the \Coctp problem with running time $\ostar(12^{\cw})$, presenting a new application of the "representative isolation" technique introduced by Bojikian and Kratsch \cite{DBLP:conf/stacs/BojikianCHK23}. We proved that the running time is tight assuming SETH. This closes the second open problem posed by Hegerfeld and Kratsch \cite{DBLP:conf/esa/HegerfeldK23}. However, one more interesting benchmark problem is still open, namely the \Fvsp problem. As mentioned in their paper, solving this problem using the \cnc technique is challenging, since a direct application would require counting edges induced by partial solutions. Another approach would be to follow the approach of Bojikian and Kratsch \cite{DBLP:journals/corr/abs-2307-14264/BojikianK23}. The difficulty of such application lays in finding a low $\bin$-rank representation of partial solutions.

Since both related techniques, namely the \cnc technique \cite{DBLP:journals/talg/CyganNPPRW22} and the rank-based approach \cite{DBLP:journals/iandc/BodlaenderCKN15}, were used to solve connectivity problems more efficiently under different structural parameters, one would ask whether the approach of Bojikian and Kratsch \cite{DBLP:journals/corr/abs-2307-14264/BojikianK23} can also be adjusted to different settings, resulting in tight bounds for connectivity problems under different parameters, where there also exists a gap between the rank of a corresponding consistency matrix, and the size of a largest triangular submatrix thereof.

\paragraph*{Acknowledgement.} 
The authors are grateful to Vera Chekan for her detailed comments on presentation, and to Falko Hegerfeld for several discussions on the lower bound presented here.

\bibliography{ref}

\end{document}